\documentclass[12pt,draftcls,onecolumn]{IEEEtran}

\usepackage{amssymb,amsmath,amsfonts,amsthm}
\usepackage{algorithm,algorithmic}
\usepackage{cite}
\usepackage{float}
\usepackage{epsfig}
\usepackage{graphicx}
\usepackage{graphics}
\usepackage{subfigure}
\usepackage{url}
\usepackage{verbatim}
\usepackage{xspace}
\usepackage[usenames,dvipsnames]{color}
\usepackage{setspace}

\floatstyle{ruled}
\newfloat{model}{H}{mod}
\floatname{model}{\footnotesize Model}
\newfloat{notatio}{H}{not}
\floatname{notatio}{\footnotesize Notation}

\usepackage{url,changebar,bm,xspace,dsfont}
\let\mathbb=\mathds 
\def\diag{\mathop{\mathrm{diag}}}  

\newtheorem{theorem}{Theorem}

\newtheorem{prop}{Proposition}
\newtheorem{assum}{Assumptions}
\newtheorem{lem}{Lemma}
\newtheorem{remark}{Remark}

\newtheorem{exam}{Example}

\ifCLASSINFOpdf
 \else
\fi

\begin{document}

%
%
%
%
\title{\LARGE \bf
Average Consensus in the Presence of Delays and Dynamically Changing Directed Graph Topologies\footnote{Preliminary results of the work in this paper have been presented in \cite{2011:Christoforos-Themis} which only included the discussions on bounded delay and did not provide all details for the proofs.}}

\author{Christoforos~N.~Hadjicostis,~\IEEEmembership{Senior Member,~IEEE}  \thanks{Christoforos~N.~Hadjicostis is with the Department of Electrical and Computer Engineering at the University of Cyprus, Nicosia, Cyprus and also with the Department of Electrical and Computer Engineering at the University of Illinois, Urbana-Champaign, IL, USA. E-mail:{\tt~chadjic@ucy.ac.cy}.} and
Themistoklis~Charalambous,~\IEEEmembership{Member,~IEEE}
        \thanks{Themistoklis~Charalambous was formerly with the Department of Electrical and Computer Engineering at the University of Cyprus, Nicosia, Cyprus. He is currently with the Automatic Control Lab, Electrical Engineering Department and ACCESS Linnaeus Center, Royal Institute of Technology (KTH), Stockholm, Sweden.  Corresponding author's address: Osquldas v\"{a}g 10, 100-44 Stockholm, Sweden. E-mail: {\tt themisc@kth.se}.}
}

\maketitle
\thispagestyle{empty}
\pagestyle{empty}

%
%
%
%
\begin{abstract}
Classical approaches for asymptotic convergence to the global average in a distributed fashion typically assume timely and reliable exchange of information between neighboring components of a given multi-component system. These assumptions are not necessarily valid in practical settings due to varying delays that might affect transmissions at different times, as well as possible changes in the underlying interconnection topology (e.g., due to component mobility). In this work, we propose protocols to overcome these limitations. We first consider a fixed interconnection topology (captured by a possibly directed graph) and propose a discrete-time protocol that can reach asymptotic average consensus in a distributed fashion, despite the presence of arbitrary (but bounded) delays in the communication links. The protocol requires that each component has knowledge of the number of its out-neighbors (i.e., the number of components to which it sends information). We subsequently extend the protocol to also handle changes in the underlying interconnection topology and describe a variety of rather loose conditions under which the modified protocol allows the components to reach asymptotic average consensus. The proposed algorithms are illustrated via examples.

\end{abstract}

%
%
%
%
\section{INTRODUCTION}

Centralized approaches in multi-component systems require the collection of measurements or other information to a central location (at possibly high communication and computational cost), the computation of quantities of interest at this central location, and then the dissemination of these quantities to (a subset of) the components. This approach is often inefficient or even unrealizable (e.g., in ad-hoc networks that do not posses explicit routing mechanisms). Cooperative distributed control/coordination algorithms and protocols have therefore received tremendous attention, especially during the last decade. Several diverse research communities (e.g., biology, physics, control, communication, and computer science) have made important contributions that have resulted in many recent advances in so called consensus-based approaches (see, for example, \cite{2007:olfati-saber_consensus}) and in distributed computation of functions of geographically dispersed data, also known as in-network computation (see, for example, \cite{2006:inNetworkComputation} and references therein).

A distributed system or network consists of a set of components (nodes) that can share information via connection links (edges), forming a directed interconnection topology (digraph). In general, the objective of a consensus problem is to have all agents agree upon a certain ({\em a priori} unknown) quantity of interest that is typically a function of some values that the nodes initially posses. When the agents (asymptotically) reach agreement to the same value, we say that the distributed system (asymptotically) reaches consensus. The problem of convergence of discrete-time consensus algorithms was initially targeted by Tsitsikis \emph{et al.} \cite{1986:tsitsiklis} and subsequently by many other researchers (see, for example, \cite{2004:Murray, 2005:wei_ren_consensus, 2005:Moreau, 2006:angeli_stability, 2008:Bliman, 2010:Franceschelli, 2011:Christoforos, 2011:cai,2012:damiano}). Typical applications include network coordination problems involving self-organization, formation of patterns, parallel processing, and distributed optimization, such as motion of mobile agents (e.g., coordination of unmanned air vehicles, autonomous underwater vehicles, or satellites) and averaging of measurements in wireless sensor networks. 

Common challenges in consensus problems include the handling of node failures (e.g., due to the draining of batteries in wireless sensor networks), transmission delays on the transfer of data between agents, packet losses in wireless communication networks, and inaccurate sensor measurements. As a result, agreement problems have been studied in networks of dynamical agents, possibly with directed information flow, under delays and/or changing topologies. It is worth pointing out that convergence of consensus algorithms can usually be established under relatively weak requirements and that consensus protocols have been relatively successful in addressing disturbances due to delays (e.g., \cite{2005:lei_fang, 2006:angeli_stability, 2008:Bliman}), packet drops (e.g., \cite{2007:Patterson,2009:Fagnani}), and changing interconnections (e.g., \cite{2005:wei_ren_consensus, 2005:Moreau}), or a combination of them (e.g., \cite{2004:Murray,2006:XiaoWang, 2010:Nedic}). What is different in this paper is that we devise a protocol that is able to overcome such limitations while reaching consensus to the exact average of the values that the nodes initially posses. More specifically, by utilizing the suggested protocol, each agent reaches the exact average of the initial values of all the agents, even in the presence of (bounded) transmission delays and changes in the interconnection topology. This is in sharp contrast with the above mentioned works in which consensus is reached to a value that is typically a function of the disturbances involved (and thus cannot be guaranteed a priori).

The \emph{average consensus} problem studied in this paper aims to have the agents reach agreement to the average of their initial values (see, for example, \cite{2003:jadbabaie_coordination, 2003:Xiao_fastlinear}). It has been shown in \cite{2004:Murray} that, under a fixed interconnection topology, average consensus can be achieved by performing a linear iteration in a distributed fashion if the interconnection topology is both strongly connected and balanced, while 
convex optimization \cite{2009:Johansson,2011:Zanella, 2012:zanellaDamiano} requires update matrices that are doubly stochastic. Even though various approaches have been proposed for forming balanced matrices (e.g., \cite{2009:Cortes,2012:Rikos}) and primitive doubly stochastic matrices (e.g., \cite{2010:Cortes, 2011:Christoforos}), which can subsequently be used for reaching average consensus, most existing schemes are not applicable in digraphs and/or fail in the presence of delays and changing interconnection topology. In particular, among the limited existing algorithms that guarantee convergence to the exact average in a digraph (e.g., \cite{2010:christoforos,2011:Franceschelli,2012:Ishii}), few of them have addressed delays and topology changes, and it is unclear how/if these techniques can be modified to overcome such disturbances.

The methodology developed in this paper is based on an algorithm suggested in \cite{2010:christoforos} that solves the average consensus problem in a static digraph using a linear iteration strategy in which each node $v_j$ distributively sets the weights on its self-link and outgoing-links to be $\frac{1}{1+\mathcal{D}_j^+}$ (where $\mathcal{D}_j^+$ is the out-degree of node $v_j$). More generally, the set of weights needs to adhere to the graph structure (i.e., be positive on each edge ---including self-edges--- and zero otherwise), but is otherwise unrestricted as long as it forms a primitive column stochastic matrix $P$. 
More generally, the set of weights needs to adhere to the graph structure (i.e., be positive on each edge -- including self-edges -- and zero otherwise), but is otherwise unrestricted as long as it forms a primitive column stochastic matrix $P$. Using the weights in matrix $P$, average consensus is reached in \cite{2010:christoforos} via \emph{ratio consensus}, i.e., two linear iterations (with appropriately chosen initial conditions) that run simultaneously so that the average can be obtained at each node by taking the ratio of the two values it maintains for each of the two iterations. An equivalent approach for gossiping algorithms was also proposed in \cite{2010:Benezit}, which is a generalization of the gossiping algorithm proposed in \cite{2003:Kempe}. The idea of ratio consensus can be traced back much earlier (see the discussion on weak convergence at the ``Bibliography and Discussion to \S\S3.1-3.2'', pp. 98, in \cite{2006:Seneta}).

In this paper, we investigate the problem of discrete-time \emph{average} consensus in a multi-component system under a directed interconnection topology in the presence of bounded delays in the communication links and changing interconnections (with communication links being added or removed, as in a mobile network setting). First, we consider a fixed topology and we devise a protocol, based on ratio consensus, where each node updates its information state by linearly combining the available (possibly delayed) information received by its neighbors using constant positive weights. We establish that, unlike other consensus approaches, this robustified version of ratio consensus, henceforth called \emph{robustified ratio consensus}, converges to the exact average of the nodes initial values, despite the presence of arbitrary but bounded time-delays. Then, we allow the communication links to change (at the same time we also allow communication delays in the network) and enhance the proposed robustified ratio consensus algorithm so that the algorithm is immune to arbitrary changing interconnection topology and delays. We show that it is possible to asymptotically reach consensus to the exact average in a distributed fashion for a network with changing communication links and delays, as long as the delays are bounded and the unions of digraphs over consecutive time intervals form strongly connected digraphs infinitely often. 


The remainder of the paper is organised as follows. In Section~\ref{preliminaries}, the notation used throughout the paper is provided, along with some background on graph theory that is needed for our subsequent development. This section also outlines the algorithm proposed in this paper. In Section~\ref{model}, we describe our model for communication link delays and our model for changing interconnection topology in the multi-agent system. In Section~\ref{delays} we consider a fixed topology and study the behavior of our algorithm in the presence of delays. In Section~\ref{switching} we consider a fixed set of nodes and allow changes in the communication links among them in order to study the behavior of our algorithm in the presence of both interconnection topology changes and delays. 
Finally, Section \ref{conclusions} summarizes the results of the paper and draws directions for future research.

%
%
%
%
\section{NOTATION AND PRELIMINARIES}\label{preliminaries}

\subsection{Notation}

The sets of real, integer and natural numbers are denoted by $\mathds{R}$, $\mathds{Z}$ and $\mathds{N}$, respectively; their nonnegative counterparts are denoted by the subscript $+$ (e.g., $\mathds{R}_{+}$). Vectors are denoted by small letters whereas matrices are denoted by capital letters. The transpose of matrix $A$ is denoted by $A^{T}$.
By $\mathbb{1}$ we denote the all-ones vector and by $I$ we denote the identity matrix (of appropriate dimensions).
A matrix whose elements are nonnegative, called nonnegative matrix, is denoted by $A \geq 0$, and a matrix whose elements are positive, called positive matrix, is denoted by $A>0$.

In multi-component systems with fixed communication links (edges), the exchange of information between components (nodes) can be conveniently captured by a digraph $\mathcal{G}(\mathcal{V}, \mathcal{E})$ of order $n$ $(n \geq 2)$, where $\mathcal{V} = \{v_1,v_2,\ldots,v_n\}$ is the set of nodes and $\mathcal{E} \subseteq \mathcal{V} \times \mathcal{V}$ is the set of edges. A directed edge from node $v_i$ to node $v_j$ is denoted by $\varepsilon_{ji} \triangleq (v_j, v_i)\in \mathcal{E}$ and represents a communication link that allows node $v_j$ to receive information from node $v_i$. A graph is said to be undirected if and only if $\varepsilon_{ji} \in \mathcal{E}$ implies $\varepsilon_{ij}  \in \mathcal{E}$. In this paper, links are not required to be bidirectional, i.e. we deal with digraphs; for this reason, we use the terms ``graph" and ``digraph" interchangeably. Note that by convention and for notational purposes, we assume that the given graph does not include any self-loops (i.e., $\varepsilon_{jj} \notin \mathcal{E}$ for all $v_j \in \mathcal{V}$) although each node $v_j$ obviously has a link (access) to its own information. A digraph is called \emph{strongly} connected if there exists a path from each vertex $v_i$ in the graph to each vertex $v_j$ ($v_j \neq v_i$). In other words, for any $v_j, v_i \in \mathcal{V}$, $v_j\neq v_i$, one can find a sequence of nodes $v_i = v_{l_1}$, $v_{l_2}$, $v_{l_3}$, $\ldots$, $v_{l_t}=v_j$ ($t\geq 2$) such that link $(v_{l_{k+1}}, v_{l_{k}}) \in \mathcal{E}$ for all $k = 1, 2, \ldots, t-1$.

All nodes that can transmit information to node $v_j$ directly are said to be in-neighbors of node $v_j$ and belong to the set $\mathcal{N}^{-}_j=\{ v_i \in \mathcal{V} \; | \; \varepsilon_{ji} \in \mathcal{E} \}$. The cardinality of $\mathcal{N}^{-}_j$, is called the \emph{in-degree} of $v_j$ and is denoted by $\mathcal{D}^{-}_{j}=\left| \mathcal{N}^{-}_j \right|$. The nodes that receive information from node $v_j$ belong to the set of out-neighbors of node $v_j$, denoted by $\mathcal{N}^{+}_j=\{ v_l \in \mathcal{V} \; | \; \varepsilon_{lj} \in \mathcal{E} \}$. The cardinality of $\mathcal{N}^{+}_j$, is called the \emph{out-degree} of $v_j$ and is denoted by $\mathcal{D}^{+}_{j}= \left| \mathcal{N}^{+}_j \right|$.

In the algorithms we will consider, we will associate a positive weight $p_{ji}$ for each $(j,i)$ such that edge $\varepsilon_{ji} \in \mathcal{E} \cup \{ (v_j, v_j) \; | \: v_j \in \mathcal{V} \}$. The nonnegative matrix $P = [p_{ji} ] \in \mathbb{R}_{+}^{n\times n}$ (with $p_{ji}$ as the entry at its $j$th row, $i$th column position) is a weighted adjacency matrix (also referred to as weight matrix) that has zero entries at locations that do not correspond to directed edges (or self-edges) in the digraph. In other words, apart from the main diagonal, the zero/nonzero structure of the adjacency matrix $P$ matches exactly the given set of links in the digraph. 
%
%
We use $x_j[k]\in \mathbb{R}$ to denote the information state of node $j$ at time step $k$. We first consider a static network where the graph connectivity remains largely invariant (as it is usually the case for distributed resources in applications, such as the power grid \cite{2005:powerGrid, 2010:christoforos}). At each time step $k$, each node $v_j$ updates its information state to $x_{j}[k+1]$ as a weighted linear combination of its own value $x_j[k]$ and the available information received by its neighbors $\{ x_{i}[k] \; | \; v_i \in \mathcal{N}^{-}_j \}$. The positive constant $p_{ji}$ captures the weight of the information inflow from agent $v_i$ to agent $v_j$. In this work, since we deal with digraphs, we assume that each node $v_j$ chooses its self-weight $p_{jj}$ and the weights $p_{lj}$ on its out-going links $v_l\in\mathcal{N}^{+}_j$. Hence, in its general form, each node updates its information state $x_j[k+1]$ according to 
\begin{eqnarray}
x_{j}[k+1]  = p_{jj} x_{j}[k] + \sum_{v_i \in \mathcal{N}^{-}_j} p_{ji}  x_{i}[k]   =  p_{jj} x_{j}[k] + \sum_{v_i \in \mathcal{N}^{-}_j} x_{j\leftarrow i}[k], \; k=0,1,2, \ldots \label{eq:1_1}
\end{eqnarray}
where $x_{j\leftarrow i}[k]\triangleq p_{ji}  x_{i}[k]$, $x_{i}[k] \in \mathbb{R}$, is the value sent to node $v_j$ by node $v_i$ at time step $k$. [Note that, since node $v_i$ chooses the weight $p_{ji}$, it is more convenient to sent $x_{j \leftarrow i}[k]$ instead of separately sending $p_{ji}$ and $x_i[k]$.] If we let $x[k]=(x_1[k] \ \ x_2[k] \ \ \ldots  \ \  x_n[k] )^T$ and $P = [p_{ji}] \in \mathbb{R}_{+}^{n\times n}$, then \eqref{eq:1_1} can be written in matrix form as
\begin{align}\label{eq:2_2}
x[k+1] =P x[k] .
\end{align}
Note that, with the exception of the diagonal entries, we have $p_{ji} =0$, $j \neq i$, if and only if $(v_j, v_i) \notin \mathcal{E}$. We say that the nodes asymptotically reach average consensus if $\lim_{k \rightarrow \infty} x_j[k] = \frac{\sum_i x_i[0]}{n}$ for all $v_j \in \mathcal{V}$. The necessary and sufficient conditions for \eqref{eq:2_2} to reach average consensus are the following \cite{2003:Xiao_fastlinear}: (a) $P$ has a simple eigenvalue  $\lambda_i(P)=1$ with left eigenvector  $\mathbb{1}^T$ and right eigenvector $\mathbb{1}$, and (b) all other eigenvalues of P ($\lambda_j(P), j\neq i$)  have magnitude less than 1 ($|\lambda_j(P)|<1$). If $P\geq 0$ (as in our case), the necessary and sufficient condition is that $P$ be a primitive doubly stochastic matrix. 

To capture dynamically changing topologies we will assume that we are given a {\em fixed} set of components $\mathcal{V} = \{ v_1, v_2, \ldots, v_n \}$ but the set of edges among them might change at various points in time. This results in a sequence of digraphs of the form $\mathcal{G}[k] = (\mathcal{V}, \mathcal{E}[k])$. Given a collection of digraphs $\mathcal{G}[1], \ldots, \mathcal{G}[m]$ (for some $m \geq 1$) of the form $\mathcal{G}[k]=(\mathcal{V}, \mathcal{E}[k])$, $k=1,2,\ldots, m$, the \emph{union digraph} is defined as $\mathcal{G}_{1,2,\ldots,m}=(\mathcal{V}, \cup_k \mathcal{E}[k])$. The collection of digraphs is said to be \emph{jointly strongly connected}, if its corresponding union graph $\mathcal{G}_{1,2,\ldots,m}$ forms a strongly connected graph. A strongly connected graph certainly emerges if at least one of the graphs in the collection is strongly connected, but it could also emerge even if none of the graphs forming the union is strongly connected.

\subsection{Ratio Consensus}

In \cite{2010:christoforos}, the average consensus problem in a digraph is solved using ratio consensus. Each node $v_j$ distributively sets positive weights on its self-link and out-going links so that the resulting weight matrix $P$ is primitive column stochastic, but not necessarily row stochastic. [Since the graph is strongly connected, it will be sufficient for node $v_j$ to choose $p_{lj} > 0$ for $v_l \in \mathcal{N}_j^+ \cup \{ v_j \}$ (zero otherwise) such that $\sum_{v_l \in \mathcal{N}_j^+ \cup \{ v_j \}} p_{lj} = 1$.] Average consensus is then reached by using this weight matrix to run two linear iterations with appropriately chosen initial conditions and by having each node take the ratio of the two values it maintains (one for each iteration). The algorithm is stated below for a specific choice of weights, which assumes that each node knows its out-degree and sets its link weights to $\frac{1}{1+\mathcal{D}_j^+}$ (this has the additional advantage of allowing broadcasts, since the transmissions $x_{l \leftarrow j}[k] \triangleq p_{lj} x_{j}[k]$ are identical for all $v_l \in \mathcal{N}_j^+ \cup \{ v_j \}$). Note, however, that the algorithm works for any set of weights that adhere to the graph structure and form a primitive column stochastic weight matrix.

\begin{lem}\label{lemma_christoforos}\cite{2010:christoforos}
Consider a strongly connected digraph $\mathcal{G}(\mathcal{V}, \mathcal{E})$. Let $y_j[k]$ and $z_j[k]$ (for all $v_j \in \mathcal{V}$ and $k=0,1,2,\ldots$) be the result of the iterations
\begin{align}\label{eq:4}
y_j[k+1]=p_{jj} y_j[k]+ \sum_{v_i \in \mathcal{N}^{-}_j} p_{ji} y_i[k] \; , \\  
z_j[k+1]=p_{jj} z_j[k]+ \sum_{v_i \in \mathcal{N}^{-}_j} p_{ji} z_i[k] \; ,
\end{align}
where $P = [p_{ji}]$ forms a primitive column stochastic matrix, and the initial conditions are $y[0]=(y_0(1) \ \ y_0(2) \ \ldots \ y_0(|\mathcal{V}|))^T\triangleq y_0$ and  $z[0]=\mathbb{1}$. Then, the protocol asymptotically converges to
$$
\displaystyle \lim_{k\rightarrow \infty} \mu_j[k]=\frac{\sum_{v_\ell \in \mathcal{V}} y_0(\ell)}{|\mathcal{V}|} \; , \;\;\forall v_j \in \mathcal{V} \; ,
$$
where
$
\displaystyle \mu_j[k]=\frac{y_j[k]}{z_j[k]} \; .
$
\end{lem}

Note that the ratio consensus in \cite{2010:christoforos} is actually a simpler version of more general algorithms that have appeared under various names in the literature (e.g., weak ergodicity property in \cite{2006:Seneta} or the push-sum algorithm in \cite{2003:Kempe}).

\subsection{Products of SIA Matrices}

A stochastic matrix $P$ is called in \cite{1963:Wolfowitz} SIA (stochastic, indecomposable, and aperiodic) if the limit $Q = \lim_{k \rightarrow \infty} P^k$ exists and has all of its columns identical. Specifically, $Q={\bf c}_P \mathbb{1}^T$ for some nonnegative vector ${\bf c}_P$. It can be shown that this definition of a SIA matrix is equivalent to the standard definitions of indecomposability and aperiodicity for stochastic matrices.\footnote{A stochastic matrix $P\in \mathbb{R}^{n \times n}$ is said to be decomposable if there exists a nonempty proper subset $\mathcal{S} \subset \{ 1, 2, \ldots, n \}$ such that $p_{ji}=p_{ij}=0$ whenever $v_i\in \mathcal{S}$ and $v_j \notin \mathcal{S}$; also, $P$ is indecomposable if it is not decomposable. A stochastic matrix $P$ is aperiodic if the Markov chain it describes is aperiodic, that is for every state $i$ there exists $k_i$ such that for all $k' \geq k_i$, the probability of being at state $i$ after $k'$ steps is greater than zero (for all $k'$) or zero (for all $k'$). Both indecomposability and aperiodicity are properties that can be checked using the structure of the digraph that is induced by the zero/nonzero structure of matrix $P$. Specifically, indecomposability follows from having a connected digraph with a single strongly connected component; for an indecomposable matrix, aperiodicity is guaranteed as long as at least one component in the strongly connected component has a self loop.} Let $A_1, A_2, \ldots , A_m$ be any square matrices of the same order. By a \emph{word} (in the $A$'s) of length $\ell \in \mathbb{N}$ we mean the product of $\ell$ $A$'s (repetitions permitted). A \emph{stochastic}, \emph{indecomposable}, and \emph{aperiodic} (SIA) matrix is a column stochastic matrix $B$ such that $\lim_{k \rightarrow \infty} B^k$ exists and has all columns the same, i.e., it is a rank one matrix of the form $c_B \mathbb{1}^T$ for some nonnegative column vector $c_B$. For the derivation of our results we make use of the theorem by Wolfowitz \cite{1963:Wolfowitz} below.

\begin{theorem} \cite{1963:Wolfowitz}
Let $\mathcal{\overline{P}} = \{ \overline{P}_1, \overline{P}_2, \ldots, \overline{P}_{m} \}$ be a collection of column stochastic matrices of the same order such that any word in the $\overline{P}$'s is stochastic, indecomposable, and aperiodic (SIA). For any $\epsilon > 0$ there exists an integer $\nu(\epsilon)$ such that any word $B=[b_{ji}]\in \mathbb{R}_{+}^{n\times n}$ (in the $\overline{P}$'s) of length $\ell \geq \nu(\epsilon)$ satisfies $\delta(B) <\epsilon$, where $\delta(B)=\max_{j} \max_{i_1,i_2} | b_{j,i_1} - b_{j,i_2} |$.
\label{Wolfowitz}
\end{theorem}





In words, Theorem~\ref{Wolfowitz} states that for large enough $\ell$, the product of $\ell$ matrices from the collection $\mathcal{\overline{P}}$ has all of its columns approximately the same. Note that the result does not mean that all matrix products converge to a single matrix of the form $c\mathbb{1}^T$; however, for large enough $\ell$, each word $B$ will take the form $c_B\mathbb{1}^T$ for some column vector $c_B$. 


%
%

%
%
\section{Modeling Delays and Switching}\label{model}

\subsection{Modeling Delays}

We first focus on the average consensus problem in the presence of bounded delays when the communication links among components are fixed and captured by an arbitrary strongly connected digraph. More specifically, the transmission on the link from node $v_i$ to node $v_j$ at time step $k$ undergoes an {\em a priori unknown} delay $\tau_{ji}[k]$, where $\tau_{ji}[k]$ is an integer that satisfies $0 \leq \tau_{ji}[k] \leq \bar{\tau}_{ji} < \infty$ (i.e., delays are bounded). The maximum delay is denoted by $\bar{\tau} = \max_{(v_j, v_i) \in \mathcal{E}} \bar{\tau}_{ji}$. We also assume that $\tau_{jj}[k]=0$, $\forall v_j \in \mathcal{V}$, at all time instances $k$ (i.e., the own value of a node is always available without delay). 
 
Under this model, the information available to node $v_j$ at time step $k$ (and which can be used to update its value to $x_j[k+1]$) comprises of its own values and all values received by its neighbors by that time, i.e., it is a subset of the values in the set $\{x_{j\leftarrow i}[s] \; | \; 0 \leq s \leq k, v_i \in \mathcal{N}^{-}_j \cup \{ v_j \} \}$ (recall that, in the digraph setting we consider, node $v_i$ selects the weight of the link $(v_j, v_i)$ and thus sends to node $v_j$ the value $x_{j\leftarrow i}[s] \triangleq p_{ji} x_i[s]$). The protocol we will employ relies on having each node $v_j$ update its information state at time step $k$ to $x_{j}[k+1]$ by combining (in a linear fashion) its own value $x_j[k]$ and the possibly delayed information received {\em at that time step} by its in-neighbors. In terms of the notation used above, this information is captured by
$
\{ x_{j\leftarrow i}[s] \; | \; 0 \leq s \leq k, \; s+\tau_{ji}[s] = k, \; v_i \in \mathcal{N}^{-}_j \cup \{ v_j \} \} \; .
$
The exact way in which this information is used will be discussed later.

\subsection{Modeling Switching}

We consider a setting where the set of components in the multi-component system is fixed to $\mathcal{V} = \{ v_1, v_2, \ldots, v_n \}$, but the (possibly directional) communication links between them are allowed to change. A convenient way of capturing this is to assume that we have a sequence of time-varying digraphs of the form $\mathcal{G}[k]=(\mathcal{V}, \mathcal{E}[k])$). This means that at each time instant $k$, each node $v_j$ has possibly different sets of in- and out-neighbors, denoted respectively by $\mathcal{N}_j^-[k]$ and $\mathcal{N}_j^+[k]$. The in-degree and out-degree of node $v_j$ are defined as $\mathcal{D}_j^-[k] = | \mathcal{N}_j^-[k] |$ and $\mathcal{D}_j^+[k] = | \mathcal{N}_j^+[k]|$, respectively.

As in the case when there is no change in the interconnection topology, each node $v_j$ is in charge of setting the weights $p_{lj}[k]$, $v_l \in \mathcal{N}_j^+[k]$, on all links to its out-neighbors. Due to the changing topology, the weight matrix will be time-varying and will be denoted by $P[k]$. What is important is for $P[k]$ to be column stochastic and to have positive weights at all links of the graph including its diagonal elements. As in Lemma~\ref{lemma_christoforos}, nodes can easily set the weights on the links to their out-neighbors to ensure column stochasticity as long as each node $v_j$ has knowledge of its out-degree $\mathcal{N}_j^+[k]$ at each time step (in such case, each node $v_j \in \mathcal{V}$ sets $p_{lj} = \frac{1}{1+\mathcal{D}_j^+[k]}$ for $v_l \in \mathcal{N}_j^+[k] \cup \{ v_j \}$). There are various ways in which the out-degree information can become available at each node (in undirected graphs this information is obviously available but even in digraphs it can become available with simple protocols that we describe in more detail later). Even if the out-degree information becomes available with some delay, the protocols we propose can still be modified to reach consensus to the exact average of the nodes' initial values.

In our analysis of changing interconnection topology, we consider two cases. \\
\noindent {\bf (i) Switching without delays:} When we have a varying interconnection topology and there exist no delays in the communication links, each node $v_j$ updates its information state at time step $k$ to $x_{j}[k+1]$ by combining its own state $x_j[k]$ and the available information received by its neighbors $\{ x_{j\leftarrow i}[k] \; | \; v_i\in \mathcal{N}^{-}_j[k] \}$ (the latter information also includes the positive weights $p_{ji}[k]$ that capture the weight of the information inflow assigned by component $v_i$ to the link $(v_j, v_i)$ at time $k$). Here, we will consider two subcases: (a) each transmitting node knows its out-degree as soon as the change takes place, and (b) each transmitting node knows its out-degree with some delay. \\
\noindent {\bf (ii) Switching with delays:} In this case, each node $v_j$ updates its information state at time step $k$ to $x_j[k+1]$ by combining its own value $x_j[k]$ and the available (possibly delayed) information $\{ x_{j\leftarrow i}[s] \; | \; 0 \leq s \leq k, \; s+\tau_{ji}[s]=k, \; v_i \in \mathcal{N}^{-}_j[s] \}$ (the latter information also includes the positive weights $p_{ji}[s]$, that component $v_i$ assigns to link $(v_j, v_i)$ at time $s$. We consider again the two cases (a) and (b) mentioned earlier where each node $v_j$ discovers its out-degree without or with delay (the out-degree with delay will be made clear in the analysis), and also consider a third case (c) in which a node $v_j$ discovers an established link with some delay.

%
%
%
%
\section{Handling Delays in a Digraph}\label{delays}

\noindent We first start with a static digraph, where each link transmission can undergo a bounded delay. We assume that each node $v_j$ chooses its self weight $p_{jj}$ and the weights $\{ p_{lj} \; | \; v_l \in \mathcal{N}_j^+ \}$ on links to its out-neighbors so that these weights are positive and satisfy $\sum_{v_l \in \mathcal{N}_j^+ \cup \{ v_j \}} p_{lj} = 1$ for all $v_j \in \mathcal{V}$ (a simple choice would be to set all of these weights to $\frac{1}{1+\mathcal{D}_j^+}$ as in Lemma~\ref{lemma_christoforos}). We employ a protocol where each node updates its information state according to the following relation:
\begin{align}\label{eq:1}
x_{j}[k+1] &=p_{jj}x_{j}[k] + \sum_{v_i \in \mathcal{N}^{-}_j} \sum_{r=0}^{\bar{\tau} } x_{j\leftarrow i}[k-r]I_{k-r,ji}[r] \; , \; k=0, 1, 2, \ldots
\end{align}
\noindent where $x_{j\leftarrow i}[k-r] \triangleq p_{ji} x_{i}[k-r]$, $x_{j}[0] \in \mathbb{R}$ is the initial state of node $v_j$, and
\begin{align}
I_{k,ji}(\tau) =
\begin{cases} 1, & \text{if $\tau_{ji}[k] =\tau$,}
\\
0, &\text{otherwise.}
\end{cases}
\label{eq:indicatorfunction}
\end{align}
In the absence of delays, we have $\tau_{ji}[k] = 0$ and the update relation \eqref{eq:1} reduces to \eqref{eq:1_1} with constant weights. We will show that if \eqref{eq:1} is employed in place of \eqref{eq:1_1} and is used to run two iterations as in Lemma~\ref{lemma_christoforos}, the resulting ratio consensus approach can still be used to calculate the exact average, despite arbitrary but bounded delays at the communication links. Essentially, we establish that \eqref{eq:1} is a ratio consensus protocol tolerant to arbitrary but bounded delays. 

\begin{assum}
For the analysis below we are given a digraph $\mathcal{G}(\mathcal{V}, \mathcal{E})$ (that represents the information exchange between agents in a multi-agent system). Each node $v_j \in \mathcal{V}$ has an initial value $y_0(j)$ and runs ratio consensus, i.e., two versions of the iteration in \eqref{eq:1}, one with initial value $y_0(j)$ and one with initial value $z_0(j)=1$. We make the following assumptions:
\begin{enumerate}
\item[(A1)] The digraph is strongly connected, and the (nonnegative) weights $p_{lj}$ are positive for $l=j$ and $(v_l, v_j) \in \mathcal{E}$, and satisfy $\sum_{l=1}^n p_{lj} = 1$ for all $v_j \in \mathcal{V}$ (so that they form a column stochastic matrix $P$). For simplicity, we will assume that each node sets the weights on the links to its out-neighbors to $p_{lj} = \frac{1}{1+\mathcal{D}_j^+}$ for $v_l \in \mathcal{N}_j^+ \cup \{ v_j \}$ (zero otherwise). \label{A1}
\item[(A2)] There exists a finite $\bar{\tau}$ that uniformly bounds the delay terms; i.e. $\tau_{ji}[k] \leq \bar{\tau}< \infty$ for all links $(v_j, v_i) \in \mathcal{E}$ for all time instants $k$. In addition, $\tau_{jj}[k] =0$ for all $v_j \in \mathcal{V}$ and all $k$. \label{A2}
\end{enumerate}
\end{assum}
Note that Assumption~(A1) is necessary for the successful operation of any distributed algorithm seeking consensus. The particular choice of weights ensures that the weight matrix $P$ is primitive column stochastic. Assumption~(A2) implies that no message is lost in the network and every agent updates its value using values from its in-neighbors at least once every $\bar{\tau}$ consecutive updates. The proof of the theorem below is developed in the remainder of this section.

\begin{theorem} \label{the:delay}
Consider a strongly connected digraph $\mathcal{G}(\mathcal{V}, \mathcal{E})$. Let $y_j[k]$ and $z_j[k]$ (for all $v_j \in \mathcal{V}$ and $k=0,1,2,\ldots$) be the result of the iterations
\begin{align}
y_{j}[k+1] &=p_{jj}y_{j}[k] + \sum_{v_i \in \mathcal{N}^{-}_j} \sum_{r=0}^{\bar{\tau} } y_{j\leftarrow i}[k-r]I_{k-r,ji}[r] \; , \label{eq:y} \\
z_{j}[k+1] &=p_{jj}z_{j}[k] + \sum_{v_i \in \mathcal{N}^{-}_j} \sum_{r=0}^{\bar{\tau} } z_{j\leftarrow i}[k-r]I_{k-r,ji}[r] \; \label{eq:z},
\end{align}
under Assumptions~(A1) and~(A2). The initial conditions are $y[0]=(y_0(1) \ \ y_0(2) \ \ldots \ y_0(|\mathcal{V}|))^T\equiv y_0$ and  $z[0]=\mathbb{1}$, and $I_{k,ji}$ is an indicator function that captures the bounded delay $\tau_{ji}[k]$ on link $(v_j, v_i)$ at iteration $k$ (as defined in \eqref{eq:indicatorfunction}, $\tau_{ji}[k] \leq \bar{\tau}$). Then, the solution to the average consensus problem can be asymptotically obtained as
$
\displaystyle \lim_{k\rightarrow \infty} \mu_j[k]=\frac{\sum_{v_\ell \in \mathcal{V}} y_0(\ell)}{|\mathcal{V}|} \; , \; \forall v_j \in \mathcal{V} \; ,
$
where $\displaystyle \mu_j[k]=\frac{y_j[k]}{z_j[k]}$.
\end{theorem}

Notice that the two iterations in the above theorem are coupled via the delays (the indicator functions $I_{k, ji}$ are the same in both iterations). Our proof is based on an augmented representation (digraph) that allows us to model the distributed system with bounded delays as described in \eqref{eq:1}, which processes packets as soon as they arrive at the destination node. We then use this augmented representation to establish that (for fixed communication topologies) the distributed ratio consensus algorithm in \eqref{eq:y}--\eqref{eq:z} will lead to asymptotic average consensus, regardless of the nature and order of the delays, as long as they are bounded. Note that the nodes are not required to know the delay of any packet or any upper bound of the delay; each node considers all the packets that arrive at that time step, by including their value in the sum.

In the augmented digraph representation, we add extra, ``virtual'' nodes and use them to model the delays. The maximum number of ``virtual" nodes for each original node is bounded by the maximum delay $\bar{\tau}$. In particular, for each node $v_j \in \mathcal{V}$ we introduce $\bar{\tau}$ ``virtual'' nodes $v_j^{(1)},v_j^{(2)}, \ldots, v_j^{(\bar{\tau})}$. At each time step $k$, virtual node $v_j^{(\tau)}$ holds the sum of the values that are destined to arrive to node $v_j$ in $\tau$ steps. The augmented digraph has $(\bar{\tau}+1)|\mathcal{V}|$ nodes and $(1+2\bar{\tau})|\mathcal{E}|$ edges. Before presenting the general case, we illustrate the construction of the augmented digraph via an example.

\begin{exam}
Consider the network of two agents exchanging information as shown in Figure~\ref{example_2nodes}. Note that that the weights $p_{11}$, $p_{12}$, $p_{21}$, and $p_{22}$ are all strictly positive, and satisfy $p_{11}+p_{21}=1$ and $p_{22}+p_{12}=1$; in the simple case presented in the introduction (and mentioned in Theorem~\ref{the:delay}), we have $p_{12} = p_{22} = 1/(1+\mathcal{D}^+_2) = 1/2$ and $p_{21} = p_{11} = 1/(1+\mathcal{D}^+_1) = 1/2$. Suppose the agents experience delays that are bounded by $2$ ($\bar{\tau}=2$). Therefore, two extra ``virtual'' nodes will be added for each node (see Figure~\ref{fig:subfigure}), depicting the states at which the delayed messages reside before reaching their destination (refer to Figure~\ref{fig:subfigure}).
\begin{figure}[h]
\begin{center}
\includegraphics[width=0.45\columnwidth]{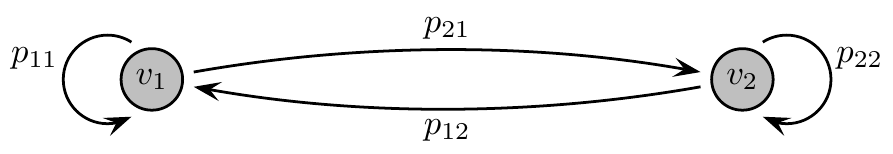}
\caption{A simple example with two nodes when the links do not experience any delays.}
\label{example_2nodes}
\end{center}
\end{figure}
\begin{figure}[h]
\subfigure[Graph representation of the network at a time instant $k=k_1$ when there exist no delays. As a result, each node uses the value sent by its neighbor directly, plus the delayed information (sent in previous time instances) that arrives at time instant $k=k_1$.]
{
\includegraphics[width=0.45\columnwidth]{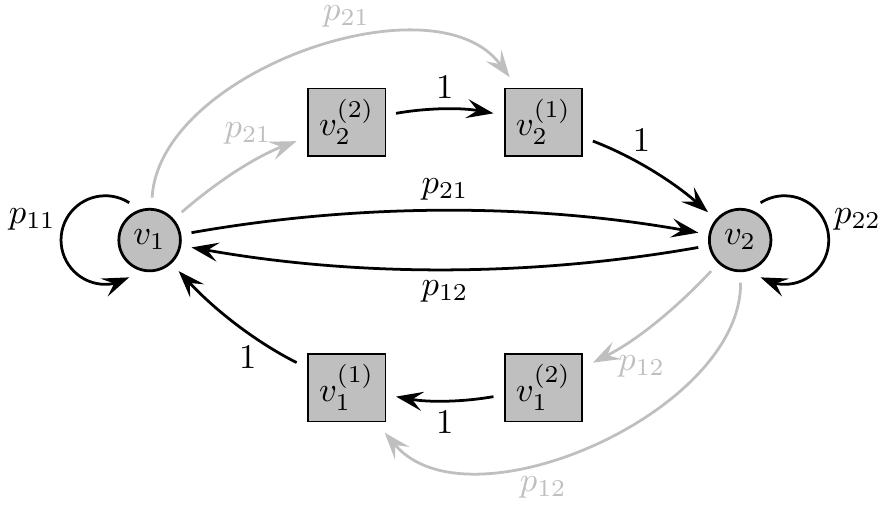}
\label{fig:subfig1}
}
\hfill
\subfigure[Graph representation of the network at a time instant $k=k_2$ for which node $v_2$ sends information to node $v_1$ with delay $\tau_{12} (k_2)=1$, while node $v_1$ sends information to node $v_2$ with delay $\tau_{21} (k_2)=2$.]
{
\includegraphics[width=0.45\columnwidth]{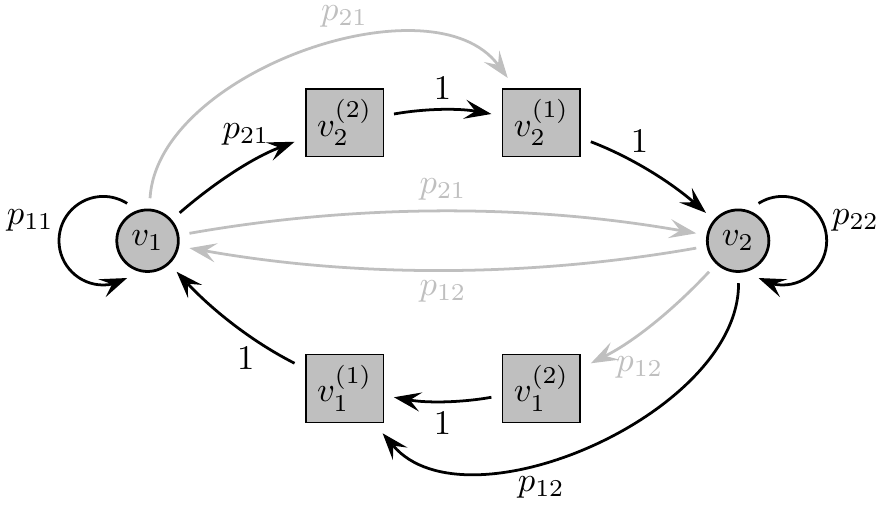}
\label{fig:subfig2}
}
\caption{
A simple example with two nodes to illustrate the modeling of delays using the proposed augmented digraph. The maximum allowable delay ($\bar{\tau}$) is 2. In Figure \subref{fig:subfig1} there exist no delays between communication links at that time instant ($k=k_1$), whereas in Figure \subref{fig:subfig2} both nodes experience delays. Active links are shown by boldface black lines in each case.
}
\label{fig:subfigure}
\end{figure}

Taking $\overline{x}[k] =(x_1[k] \ \  x_2[k]  \ \ x_{1}^{(1)}[k] \  \ x_{2}^{(1)}[k] \  \ x_{1}^{(2)}[k]  \  \ x_{2}^{(2)}[k] )^T$, the iteration in the augmented digraph can be written as
$
\overline{x}[k+1]=\overline{P}[k]\overline{x}[k],
$
where $\overline{P}[k]$ depends on the nature of delays. For example, when there are no delays in the network, say at time instant $k=k_1$, the network is represented by Figure~\ref{fig:subfigure}\subref{fig:subfig1} and the weight matrix is given by $\overline{P}[k_1]$ below. Similarly, if at time instant $k=k_2$ there is a delay of 2 at link $(v_2, v_1)$ and a delay of 1 at link $(v_1, v_2)$ (the network is shown in Figure \ref{fig:subfigure}\subref{fig:subfig2}), the matrix representation is given by $\overline{P}[k_2]$ below.
$${\small
\begin{array}{cc}
\overline{P}[k_1]=
\left( \begin{array}{cccccc}
    p_{11} & p_{12} & 1 & 0 & 0 & 0 \\
    p_{21} & p_{22} & 0 & 1 & 0 & 0 \\
    0 & 0 & 0 & 0 & 1 & 0  \\
    0 & 0 & 0 & 0 & 0 & 1 \\
    0 & 0 & 0 & 0 & 0 & 0 \\
    0 & 0 & 0 & 0 & 0 & 0 \\
  \end{array} \right) \; ,
&
\overline{P}[k_2]=
\left( \begin{array}{cccccc}
    p_{11} & 0 & 1 & 0 & 0 & 0 \\
    0 & p_{22} & 0 & 1 & 0 & 0 \\
    0 &  p_{12} & 0 & 0 & 1 & 0  \\
    0 & 0 & 0 & 0 & 0 & 1 \\
    0  & 0 & 0 & 0 & 0 & 0 \\
    p_{21} & 0 & 0 & 0 & 0 & 0 \\
  \end{array} \right).
\end{array}}
$$
\end{exam}
In the general case, in a network of $n=|\mathcal{V}|$ nodes, we introduce $\bar{\tau} n$ nodes (for a total of $(\bar{\tau}+1)n$ nodes) so that
$
\overline{x}[k+1]=\overline{P}[k]\overline{x}[k] ,
$
where
\begin{align}{\small
\label{Ak}
\overline{P}[k] \triangleq
\left( \begin{array}{ccccc}
    P_0[k] & I_{n\times n} & 0 & \cdots & 0 \\
    P_1[k] & 0 & I _{n\times n}&  \cdots &  0 \\
    \vdots & \vdots & \vdots & \ddots & \vdots \\
    P_{\bar{\tau}-1}[k] & 0 & 0 & \cdots & I_{n\times n} \\
    P_{\bar{\tau}}[k] & 0 & 0 & \cdots & 0 \\
  \end{array} \right),}
\end{align}
with $\overline{x}[k]=\left( x^T[k] \ \ x^{(1)}[k] \ \ldots \ x^{(\bar{\tau})}[k]  \right)^T$ and $x^{(r)}[k]=\left(x_{1}^{(r)}[k]  \ \ldots \  x_{n}^{(r)}[k]   \right)$, $r=1,2, \ldots \bar{\tau}$. Note that $P_0[k], P_1[k], \ldots, P_{\bar{\tau}}[k]$ are appropriately defined nonnegative matrices that depend on the link delays that are experienced by messages sent at time $k$. Specifically, $P_r[k]$ is a matrix associated only with the links of the graph for which the message was delayed by $r$ steps at time step $k$, and satisfies
\begin{align*}
P_{r}[k](j,i) =
\begin{cases} P(j,i) , & \text{if $\tau_{ji}[k] =r$, \ $(j,i)\in\mathcal{E}$,}
\\
0, &\text{otherwise.}
\end{cases}
\end{align*}
Note that, for each $(j,i) \in \mathcal{E}$, only one of $P_0[k](j,i)$, $P_1[k](j,i)$, ..., $P_{\bar{\tau}}[k](j,i)$ is nonzero and is equal to $P(j,i)$. Thus, we also have
\begin{equation}
P=\sum_{r=0}^{\bar{\tau}}P_r[k] \; , \quad k = 0, 1, 2, \ldots \label{eqPmatrix}
\end{equation}
Matrix $\overline{P}[k]$ may take at most $(\bar{\tau}+1)^{|\mathcal{E}|}$ matrix values, where $(\bar{\tau}+1)$ is the total number of states (``virtual'' and original) for each node $v_j$. Specifically, if there exists an edge $(v_j, v_i)$ in the original digraph, then that edge also exists in the augmented digraph along with edges $(v_j^{(1)}, v_i)$, $(v_j^{(2)}, v_i)$, $\ldots$, $(v_j^{(\bar{\tau})}, v_i)$, and also edges $(v_j, v_j^{(1)})$, $(v_j^{(1)}, v_j^{(2)})$, $\ldots$, $(v_j^{(\bar{\tau}-1)}, v_j^{(\bar{\tau})})$. However, among the $(\bar{\tau}+1)$ entries of $\overline{P}[k]$ corresponding to the edge $(v_j, v_i)$ only one of them could be nonzero (and equal to $p_{ji}$); the others will be zero. In the sequel we do not require the matrix $\overline{P}[k]$ to be known at each time step $k$; what we utilize is that $\overline{P}[k]$ will be a matrix from a finite set of possible matrices $\overline{\mathcal{P}}$, which have certain useful properties.

\begin{prop}\label{proposition1}
Let $\overline{\mathcal{P}} = \{\overline{P}_1, \overline{P}_2, \ldots, \overline{P}_{(\bar{\tau}+1)^{|\mathcal{E}|}}\}$ be the set of all possible $\overline{P}[k]$ as defined in \eqref{Ak}. Then, for integer $\ell$, $\ell \geq \bar{\tau}+1$, any $\ell$-length word $B = \overline{P}[k+\ell]\overline{P}[k+\ell-1]\ldots \overline{P}[k+1]$ is SIA. Moreover, for $\ell \geq n(\bar{\tau}+1)$, the first $n$ rows of matrix $B$ will be positive with minimum entry greater or equal to $c_{\min} \equiv \left ( \frac{1}{\mathcal{D}_{\max}^+} \right )^{n(\bar{\tau}+1)}$, where $\mathcal{D}_{\max}^+ = \max_{v_j \in \mathcal{V}} \mathcal{D}_j^+$.
\end{prop}

\begin{proof}[Proof of Proposition \ref{proposition1}] In order to prove that $B=\overline{P}[k+\ell]\ldots \overline{P}[k+2] \overline{P}[k+1]$ is SIA, we have to show that it is (i) column stochastic, (ii) indecomposable, and (iii) aperiodic. \\
\noindent (i) \textit{Column Stochasticity}: This is easy to see as it is equivalent to proving that the product of two or more column stochastic matrices of the same order is also a column stochastic matrix (the result follows easily by induction and is standard).

\noindent (ii) \textit{Indecomposability}: We argue indecomposability for $\ell \geq \bar{\tau}+1$ (the result also holds for any $0 \leq \ell < \bar{\tau}+1$ but we do not discuss the proof here due to space limitations). Write matrix $B$ in block form as 
\begin{align*}{\small
B = \left( \begin{array}{ccccc}
    B_{0,0} & B_{0,1} & B_{0,2} & \cdots & B_{0,\bar{\tau}} \\
    B_{1,0} & B_{1,1} & B_{1,2} & \cdots & B_{1,\bar{\tau}} \\
    \vdots & \vdots & \vdots & \ddots & \vdots \\
    B_{\bar{\tau}-1,0} & B_{\bar{\tau}-1,1} & B_{\bar{\tau}-1,2} & \cdots & B_{\bar{\tau}-1,\bar{\tau}} \\
    B_{\bar{\tau},0} & B_{\bar{\tau},1} & B_{\bar{\tau},2} & \cdots & B_{\bar{\tau},\bar{\tau}} \\
  \end{array} \right) ,}
\end{align*}
where all blocks are nonnegative matrices of size $n \times n$. We will argue that (i) the zero/nonzero structure of $B_{0,0}$ corresponds to a graph that is strongly connected, and (ii) each of $B_{0,0}$, $B_{0,1}$, $B_{0,2}$, ..., $B_{0,\bar{\tau}}$ has strictly positive entries on its diagonal. These two facts establish that the graph that corresponds to the zero/nonzero structure of the overall matrix $B$ has the following property: (i) any pair of non-virtual nodes (i.e., the top $n$ nodes) can be connected via a directed path (that can actually involve only non-virtual nodes); (ii) all other (virtual) nodes have an outgoing link to at least one of the non-virtual nodes. Therefore, the set of non-virtual nodes is part of a strongly connected component; this component could potentially involve other (virtual) nodes in the graph, but no other strongly connected component exists. Thus, matrix $B$ is indecomposable.

For fact (i), we need to explain why $B_{0,0}$ corresponds to a graph of $n$ nodes that is strongly connected. It is not hard to see that one can write
$$
B_{0,0} = (\Pi_{l=2}^{\ell} P_0[k+l]) P_0[k+1] + (\Pi_{l=3}^{\ell} P_0[k+l]) P_1[k+1] + ... + (\Pi_{l=\bar{\tau}+2}^{\ell} P_0[k+l]) P_{\bar{\tau}}[k+1] + E_{0,0}
$$
where $\Pi_{l=l_1}^{l_2} A[l] \equiv A[l_2] A[l_2-1] ... A[l_1]$ ($\Pi_{l=l_1}^{l_2} A[l] \equiv I$ for $l_2 = l_1-1$ and zero otherwise) and $E_{0,0}$ is a nonnegative matrix (that can be expressed as the sum of various products of the nonnegative\footnote{The fact that the blocks are nonnegative is important because it means that nonzero entries created by some products cannot be cancelled by nonzero entries of other products.} blocks composing the $\overline{P}$ matrices). Since the diagonal elements in matrix $P_0[k+l]$ (for $l=1,2, ..., \ell$ are strictly positive, we know that the diagonals of each product $\Pi_{l=m}^{\ell} P_0[k+l]$, $m = 2, 3, ..., \bar{\tau}+2$, will be strictly positive and thus the elements of each term $(\Pi_{l=m}^{\ell} P_0[k+l]) P_{m-2}[k+1]$ will be positive at the locations where $P_{m-2}[k]$ is positive. Thus, from the expression for $B_{0,0}$ above, the zero/nonzero structure of $B_{0,0}$ corresponds to a graph of $n$ nodes that includes all the edges in $\sum_{r=0}^{\bar{\tau}} P_r[k+1] = P$ (recall \eqref{eqPmatrix}); thus, all edges in the original graph are included and, since the original graph is strongly connected, $B_{0,0}$ corresponds to a graph that is strongly connected.

For fact (ii), we need to explain why each $B_{0,r}$, $r=0, 1, ..., \bar{\tau}$, has strictly positive diagonal entries. For $r=0$, this follows for the discussion above. For $r=1,2, ..., \bar{\tau}$, we can also write
$$
B_{0,r} = (\Pi_{l=r+2}^{\ell} P_0[k+l]) P_0[k+1+r] + E_{0,r} \; , 
$$
where $E_{0,r}$ is again a nonnegative matrix that can be expressed as the sum of various products of nonnegative matrices. Since the diagonal elements in matrix $P_0[k+l]$ (for $l=1,2, ..., \ell$) are strictly positive, we know that the diagonals of each $B_{0,r}$ will be strictly positive.

\noindent (iii) \textit{Aperiodicity}: Since the graph corresponding to $B$ is indecomposable, aperiodicity is easily established due to the fact that the diagonal entries that correspond to the original (non-virtual) nodes in the strongly connected component are nonzero (it is sufficient for at least one of them to be nonzero).

To prove the second part of the proposition (i.e., for $\ell \geq n(\bar{\tau}+1)$, the first $n$ rows of matrix $B$ will be positive with minimum entry greater or equal to $c_{\min} \equiv \left ( \frac{1}{\mathcal{D}_{\max}^+} \right )^{n(\bar{\tau}+1)}$), notice that for $\ell = n(\bar{\tau}+1)$ we can write $B$ as 
$$
B = \underbrace{B_{i_n} B_{i_{n-1}} ...B_{i_2}}_{B'} B_{i_1} \; ,
$$
where each $B_{i_m}$ is the product of $\bar{\tau}+1$ consecutive $\overline{P}$, i.e.,
$$
B_{i_m} = \overline{P}[k+m(\bar{\tau}+1)] \ldots \overline{P}[k+(m-1)(\bar{\tau}+1)+2] \overline{P}[k+(m-1)(\bar{\tau}+1)+1] \; .
$$
From the discussion on indecomposability, we know that each $B_{i_m}$ has blocks $B^{(i_m)}_{0,0}$, $B^{(i_m)}_{0,1}$, ..., $B^{(i_m)}_{0,\bar{\tau}}$ such that the zero/nonzero structure of $B^{(i_m)}_{0,0}$ corresponds to a graph that includes the original strongly connected graph of size $n$ and has a positive diagonal. Thus, the product of $n-1$ such blocks will result in a strictly positive diagonal block for matrix $B'$. An additional multiplication by $B_{i_1}$ on the right, will ensure that each of the top $\bar{\tau}+1$ blocks of matrix $B$ will be strictly positive. Since $B$ involves the product of $n(\bar{\tau}+1)$ nonnegative matrices $\overline{P}$ (whose minimum nonzero entry is $\frac{1}{\mathcal{D}^+_{\max}}$), the minimum entry in $B$ will be greater or equal to $c_{\min}$.

For $\ell = n(\bar{\tau}+1)+1$, we have a matrix product of the form $B \overline{P}[k+1]$, where $B$ is the product of $n(\bar{\tau}+1)$ matrices $\overline{P}$ (thus, its top $n$ rows are strictly positive with minimum entry $c_{\min}$). Since $\overline{P}[k+1]$ is a column stochastic matrix, we can easily conclude that matrix $B \overline{P}[k+1]$ will also have its top $n$ rows positive with minimum entry $c_{\min}$. The claim in the second part of the proposition (that, for $\ell \geq n(\bar{\tau}+1)$, the first $n$ rows of matrix $B$ will be positive with minimum entry greater or equal to $c_{\min}$), then follows easily by induction.
\end{proof}

\begin{proof}[Proof of Theorem~\ref{the:delay}]
If we use the augmented graph representation with initial conditions $\bar{y}[0] = [y^T_{0} \; 0 \; 0 \; \ldots \; 0]^T$ and $\bar{z}[0] = [\mathbb{1}^T \; 0 \; 0 \; \ldots \; 0]^T$, we can write
\begin{align*}
\bar{y}[k] = \underbrace{\overline{P}[k] \ldots \overline{P}[2]\overline{P}[1]}_{B_k} \bar{y}[0] \; , \\
\bar{z}[k] = \underbrace{\overline{P}[k] \ldots \overline{P}[2]\overline{P}[1]}_{B_k} \bar{z}[0] \; .
\end{align*}
By Proposition~1 and Wolfowitz's Theorem, we know that for any $\epsilon >0$, the resulting word ${B_k}$ satisfies (for $k\geq\nu(\epsilon)$) ${B_k} = c_{B_k}\mathbb{1}^{T} + E_{k}$, where $c_{B_k}$ is an appropriate nonnegative vector, and $E_{k}$ is an error matrix with entries with absolute value smaller than $\epsilon/2$ (i.e., $|E_k(j,i)| < \epsilon/2$ for all $j,i$). Taking $k$ to also satisfy $k \geq n(\bar{\tau}+1)$, it follows from Proposition~1 that each of the first $n$ entries of $c_{B_k}$ (i.e., the entries that correspond to non-virtual nodes) will be greater than $c_{\min}$. Without loss of generality, we take $\epsilon < 2c_{\min}$ in the remainder of this discussion.

With the above notation at hand, we have for $v_j \in \mathcal{V}$
\begin{eqnarray*}
\mu_j[k] \triangleq \frac{\bar{y}_j[k]}{\bar{z}_j[k]} = \frac{B_k(j,:) \bar{y}[0]}{B_k(j,:) \bar{z}[0]} 
 = \frac{c_{B_k}(j) ({\mathbb 1}^T + e^T_k) \bar{y}[0]}{c_{B_k}(j) ({\mathbb 1}^T + e^T_k) \bar{z}[0]}
 = \frac{({\mathbb 1}^T + e^T_k) \bar{y}[0]}{({\mathbb 1}^T + e^T_k) \bar{z}[0]} \; ,
\end{eqnarray*}
where $c_{B_{k}}(j)$ is the $j$th element of vector $c_{B_{k}}$, and $e^T_k = E_k(j,:)$ is the $j$th row of matrix $E_k$ and satisfies $e_{\max}(k) \equiv \max_i \{ |e_k(i)| \} < \epsilon/2$.

Since $\bar{z}[0] = \mathbb{1} \geq 0$ (elementwise), the denominator of the above expression can be bounded 
$$
n (1 - e_{\max}(k)) \leq ({\mathbb 1}^T + e^T_k) \bar{z}[0] \leq n (1 + e_{\max}(k)) \; .
$$
Similarly, assuming that $\sum_l \bar{y}_l[0] = \sum_l y_l[0] > 0$ (when $\sum_l \bar{y}_l[0] = \sum_l y_l[0] < 0$ or $\sum_l \bar{y}_l[0] = \sum_l y_l[0] = 0$ we can apply a similar analysis), we can bound the numerator of the above expression as 
$$
\Sigma_y - e_{\max}(k) \Sigma_{|y|} \leq ({\mathbb 1}^T + e^T_k) \bar{y}[0] \leq \Sigma_y + e_{\max}(k) \Sigma_{|y|} \; ,
$$
where $\Sigma_y = \sum_l \bar{y}_l[0] = \sum_l y_l[0]$ and $\Sigma_{|y|} = \sum_l | \bar{y}_l[0]|
= \sum_l | y_l[0]|$. Putting the above inequalities together, we obtain
$$
\frac{\Sigma_y - e_{\max}(k) \Sigma_{|y|}}{n (1 + e_{\max}(k))} \leq \mu_j[k] \triangleq \frac{\bar{y}_j[k]}{\bar{z}_j[k]} \leq \frac{\Sigma_y + e_{\max}(k) \Sigma_{|y|}}{n (1 - e_{\max}(k))} \; ,
$$
which can be relaxed (after some algebraic manipulations) to
$
\mu^* - M_k \leq \mu_j[k] \leq \mu^* + M_k
$, 
where $\mu^* = \frac{\sum_l y_l[0]}{n}$ is the exact average and $M_k = \mu^* \frac{(\Sigma_y+\Sigma_{|y|}) e_{\max}(k)}{\Sigma_y(1 - e_{\max}(k))}$. By Wolfowitz theorem, we can take $k$ as large as necessary to make $M_k$ arbitrarily small (by ensuring that $\epsilon$ and thus $e_{\max}(k)$) is as small as desired).
\end{proof}

\begin{exam}
Consider the directed network on the left of Figure \ref{example_directedgraph} where each node $v_j$ chooses its self-weight and the weights of its outgoing links to be $(1+\mathcal{D}_j^{+})^{-1}$ so that the weight matrix $P$ is primitive column stochastic as shown on the right of the figure. Each node $v_j$ updates its information state $x_j[k]$ using equation \eqref{eq:4}, so that the information state for the whole network is given by $x[k+1]=Px[k]$. We first use the update formula \eqref{eq:1_1} with
$
y[0]=(-1 \ \  2 \ \ 3 \ \ 4 \ \ 2)^T \equiv y_0
$
and no delays ($\bar{\tau}=0$). Since the update matrix is column stochastic, iteration \eqref{eq:4} for this network converges, but not necessarily to the average (as shown in Figure~\ref{directed_1} the nodes do not even reach consensus). As suggested in \cite{2010:christoforos}, by simultaneously running two iterations $y[k]$ and $z[k]$ (using the weights in matrix $P$) with initial conditions $y[0]=y_0$ and $z[0]=\mathbb{1}$, respectively, then average consensus is asymptotically reached for the ratio $y_j[k]/z_j[k]$ (see Figure \ref{directed_2}).

\begin{figure}[h]
\hspace*{1cm}\includegraphics[width=0.40\columnwidth]{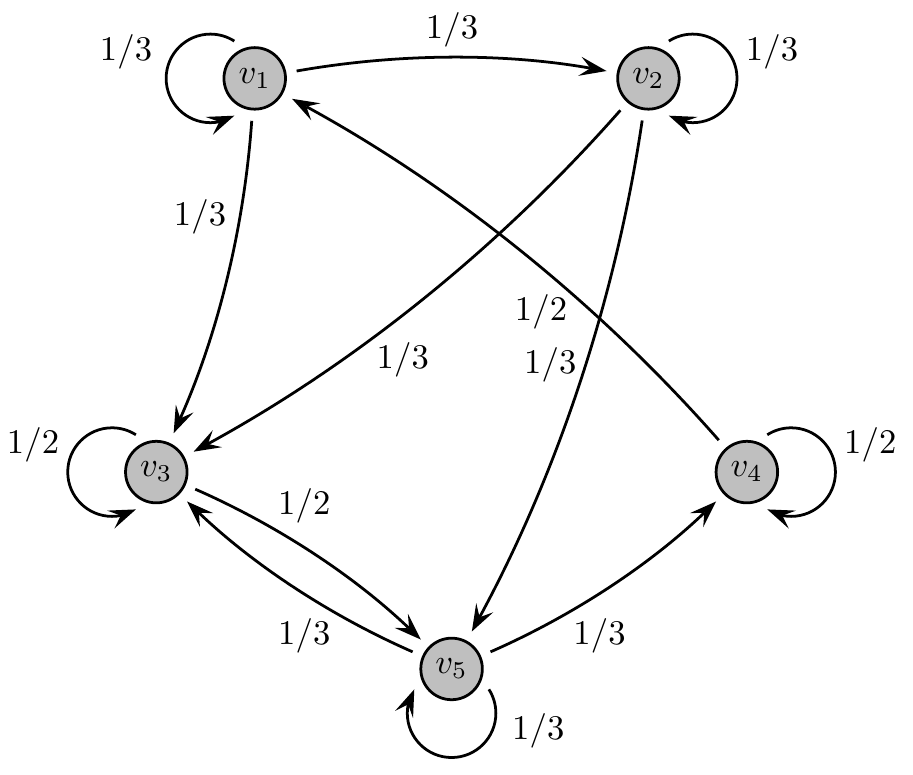}

\vspace*{-4cm}

\hspace*{9cm}~$P=
\left( \begin{array}{ccccc}
    1/3 & 0 & 0 & 1/2 & 0 \\
    1/3  & 1/3  & 0 & 0 & 0 \\
    1/3  & 1/3  & 1/2 & 0 & 1/3 \\
    0 & 0  & 0 & 1/2 & 1/3 \\
    0 & 1/3 & 1/2 & 0 & 1/3 \\
  \end{array} \right)$

\vspace*{1.5cm}

\caption{A simple digraph of five nodes when the links do not experience any delays.}
\label{example_directedgraph}
\end{figure}

\begin{figure}[h]
\subfigure[The iteration of \eqref{eq:1_1} for the network in Figure \ref{example_directedgraph} does not lead to average consensus (not even consensus) for the digraph, since $P$ is not a doubly stochastic matrix.]
{
\includegraphics[width=0.46\columnwidth]{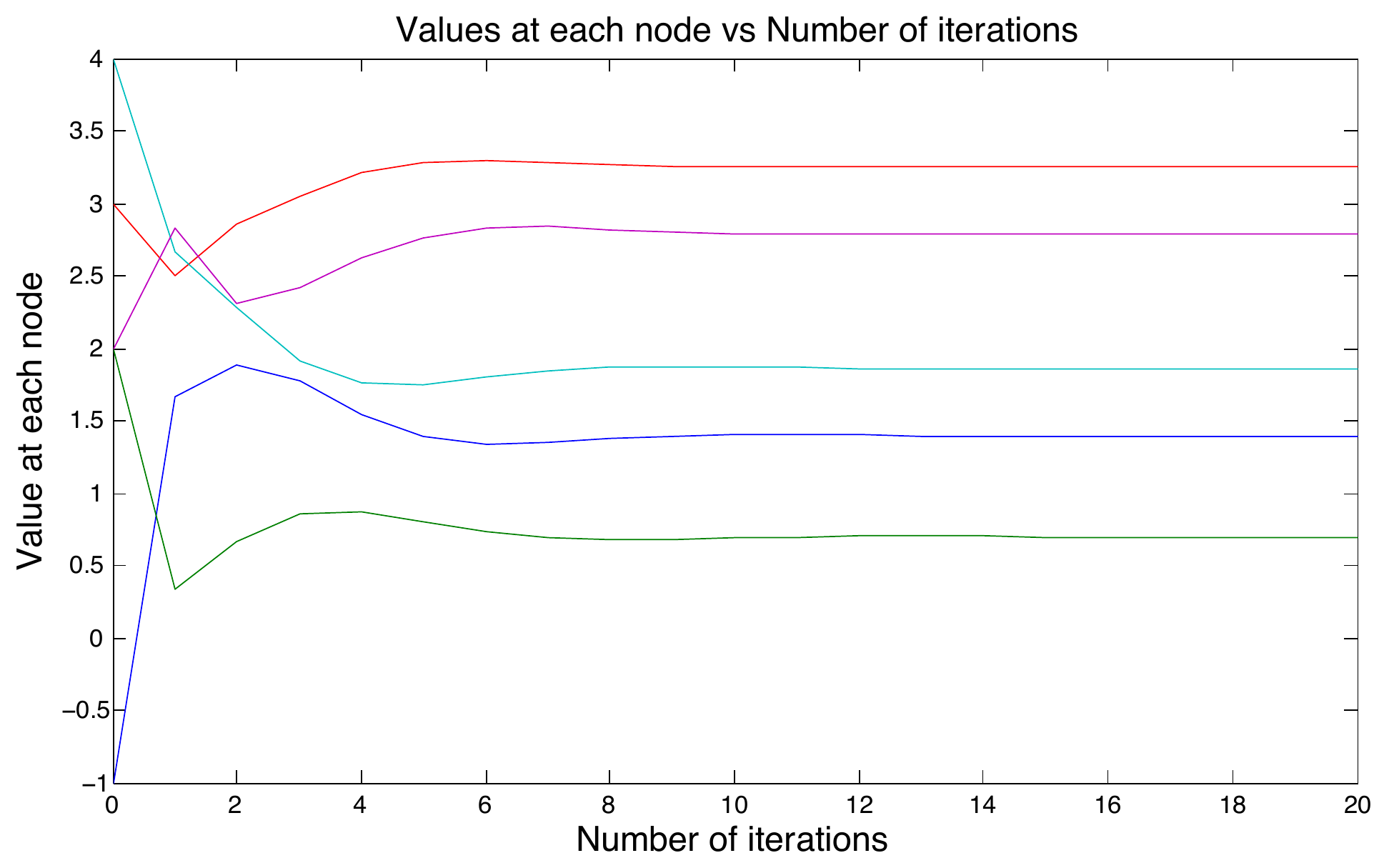}
\label{directed_1}
}
\hfill
\subfigure[By running ratio consensus in \eqref{eq:y}--\eqref{eq:z} with appropriate initial conditions, average consensus is reached (no delays are introduced yet, i.e., $\bar{\tau}=0$).]
{
\includegraphics[width=0.46\columnwidth]{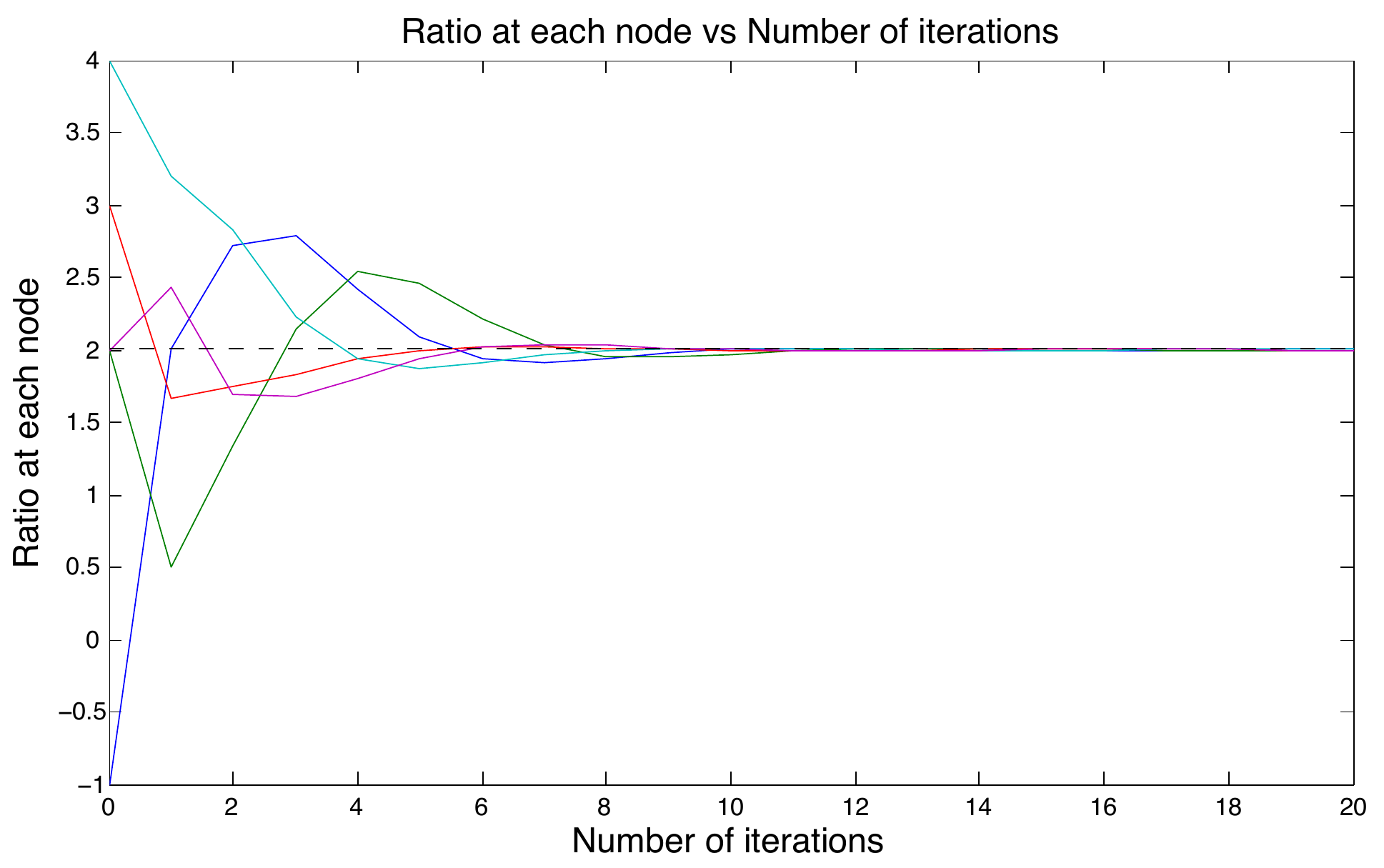}
\label{directed_2}
}
\caption{Iteration \eqref{eq:4} converges but does not reach consensus (left). By simultaneously running two iterations $y[k]$ and $z[k]$ (using the weights in matrix $P$) with initial conditions $y[0]=y_0$ and $z[0]=\mathbb{1}$, respectively, then average consensus is asymptotically reached for the ratio $y_j[k]/z_j[k]$ (right).}
\label{fig:illustration}
\end{figure}

We now consider delays by taking the maximum allowable delay to be $\bar{\tau} = 5$. At each link at each time instant, the delay is an integer in $\{ 0, 1, 2, ..., 5 \}$ (in our simulations each possible delay is chosen with probability $1/6$). If we run our update formula (as in \eqref{eq:1}) for the network in Figure \ref{example_directedgraph} with weights $P$ and
$x[0]=y_0$,
the algorithm does not converge (see Figure \ref{directed_4}). However, if we run ratio consensus in \eqref{eq:y}--\eqref{eq:z} with initial conditions $y[0]=y_0$ and $z[0]=\mathbb{1}$ respectively, then average consensus is asymptotically reached for the ratio $y_j[k]/z_j[k]$ (Figure \ref{fig:directed_3}). This demonstrates the validity of our theoretical analysis, both in the sense that each of the individual iterations does not convergence and also in the sense that the ratios converge to the average of the initial values.

\begin{figure}[h]
\subfigure[The update rule \eqref{eq:1} does not converge for the digraph, due to the presence of delays.]
{
\includegraphics[width=0.46\columnwidth]{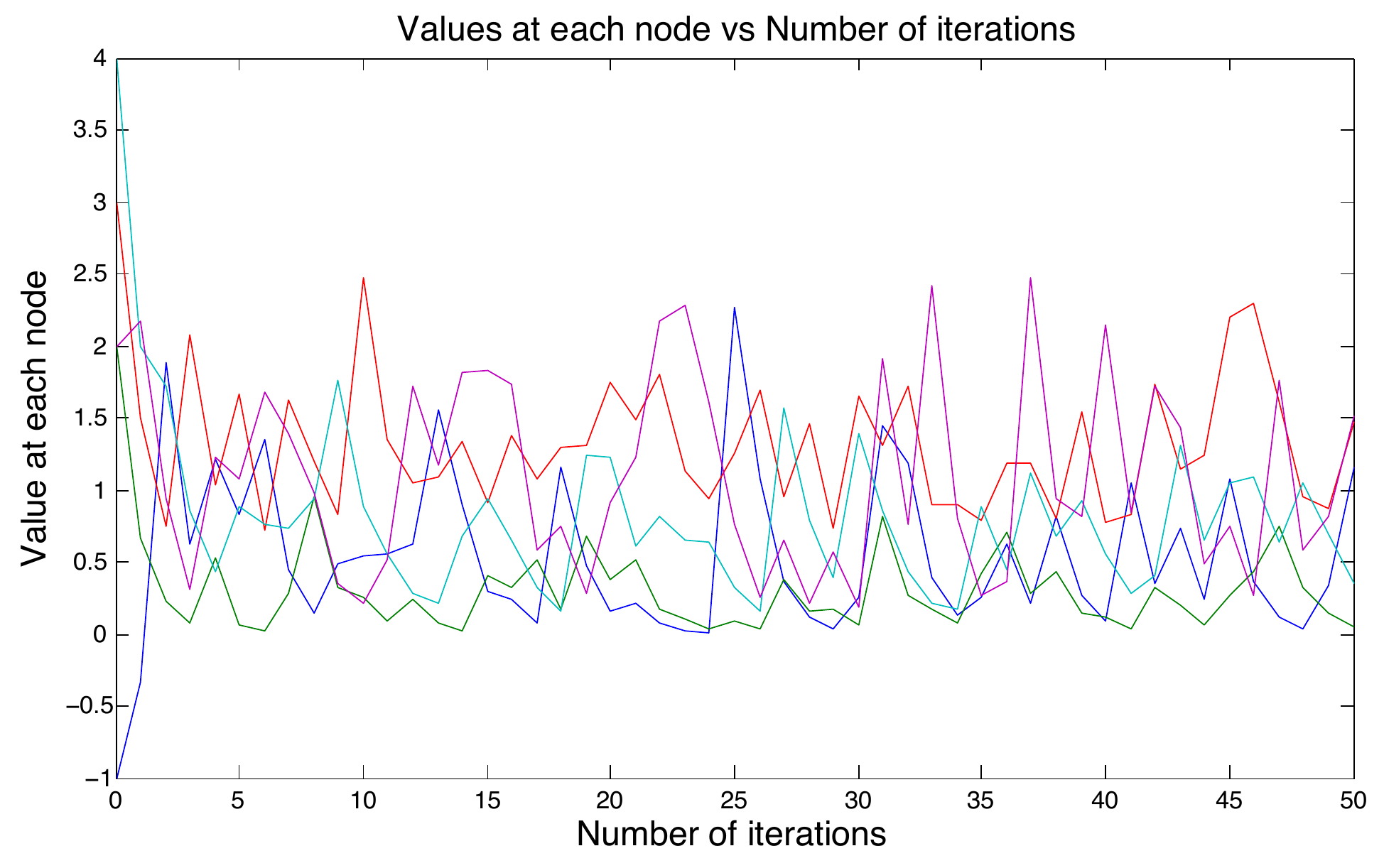}
\label{directed_4}
}
\hfill
\subfigure[By running ratio consensus in \eqref{eq:y}--\eqref{eq:z} with appropriate initial conditions, we can reach average consensus even in the presence of delays.]
{
\includegraphics[width=0.46\columnwidth]{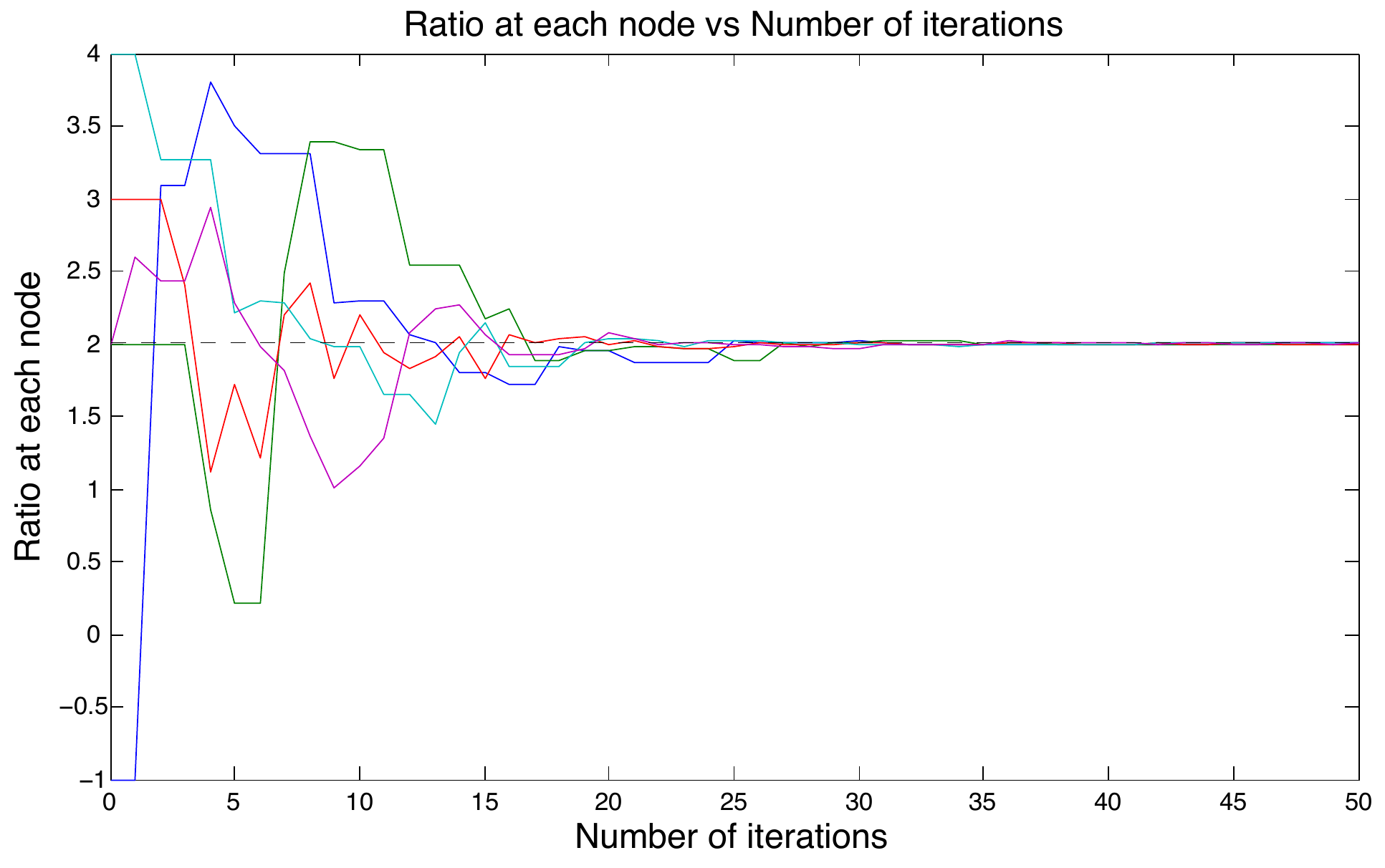}
\label{fig:directed_3}
}
\caption{The update formula in \eqref{eq:1} for the network on the left of Figure~\ref{example_directedgraph} with weights $P$ and $x[0]=y_0$, does not converge (left); however, if we run ratio consensus in \eqref{eq:y}--\eqref{eq:z} with initial conditions $y[0]=y_0$ and $z[0]=\mathbb{1}$ respectively, then average consensus is asymptotically reached for the ratio $y_j[k]/z_j[k]$ (right).
}
\label{fig:delayed_compare}
\end{figure}

It is obvious from the simulations that the convergence speed of the algorithm depends on the delays (e.g., longer delays should result in slower convergence of the algorithm). Our discussion did not characterize the worst-case combination of delays but, nevertheless, the final average value is not affected by the particular realization of delays.

\end{exam}

To gain additional insight into the problem, we also consider the convergence of node $1$ under (i) different upper bounds in delays (see left of Figure~\ref{fig:delayed_comparenode1} where delays are equally likely as before), and (ii) varying network size (see right of Figure~\ref{fig:delayed_comparenode1} where random geometric graphs of different sizes are used and $\bar{\tau}=5$ with delays being equally likely as before). It is obvious from the simulations that the convergence speed of the algorithm depends on the delays (e.g., longer delays result in slower convergence). Nevertheless, for fixed $\bar{\tau}$ it appears that the size of the network has no effect on the convergence time (at least for geometric graphs).
\begin{figure}[h]
{
\includegraphics[width=0.51\columnwidth]{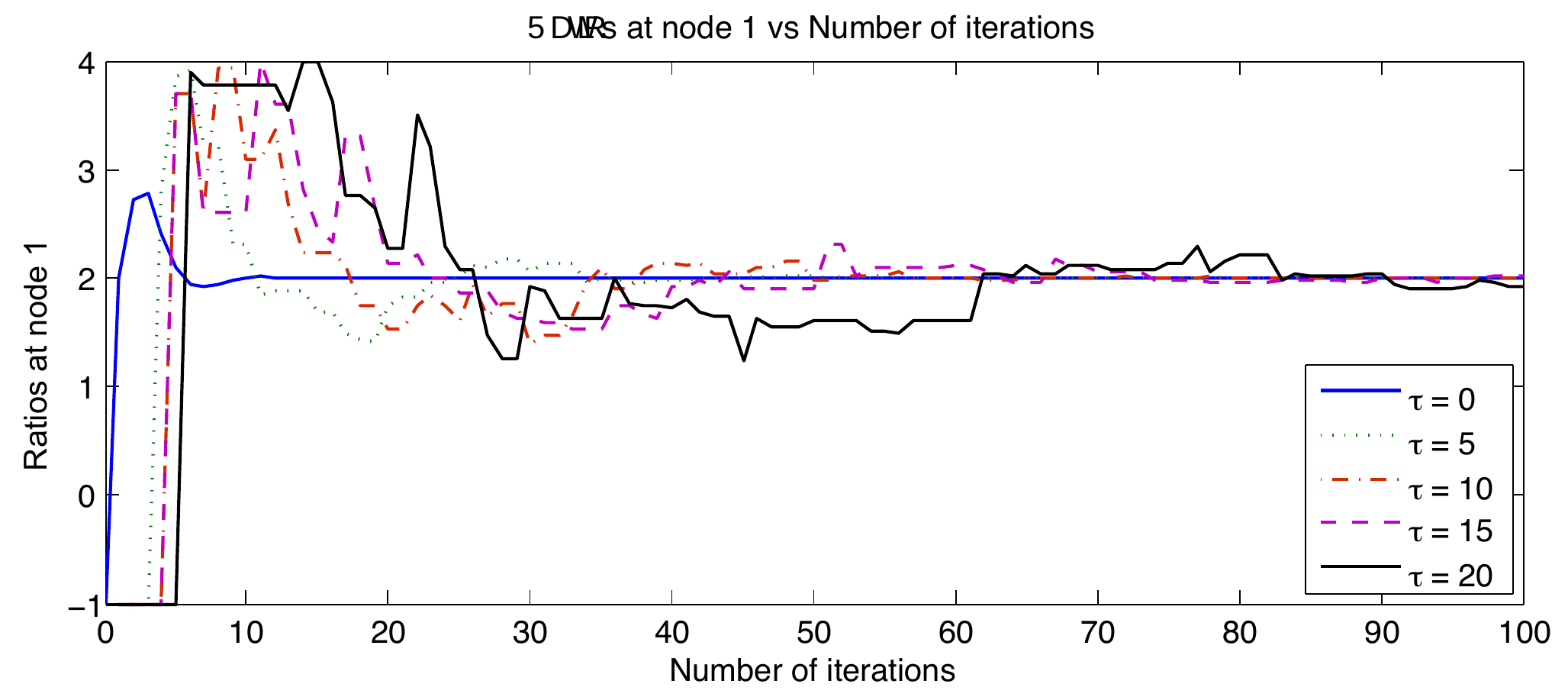}
\label{fig:comparison-delay}
}
\hfill
{
\includegraphics[width=0.44\columnwidth]{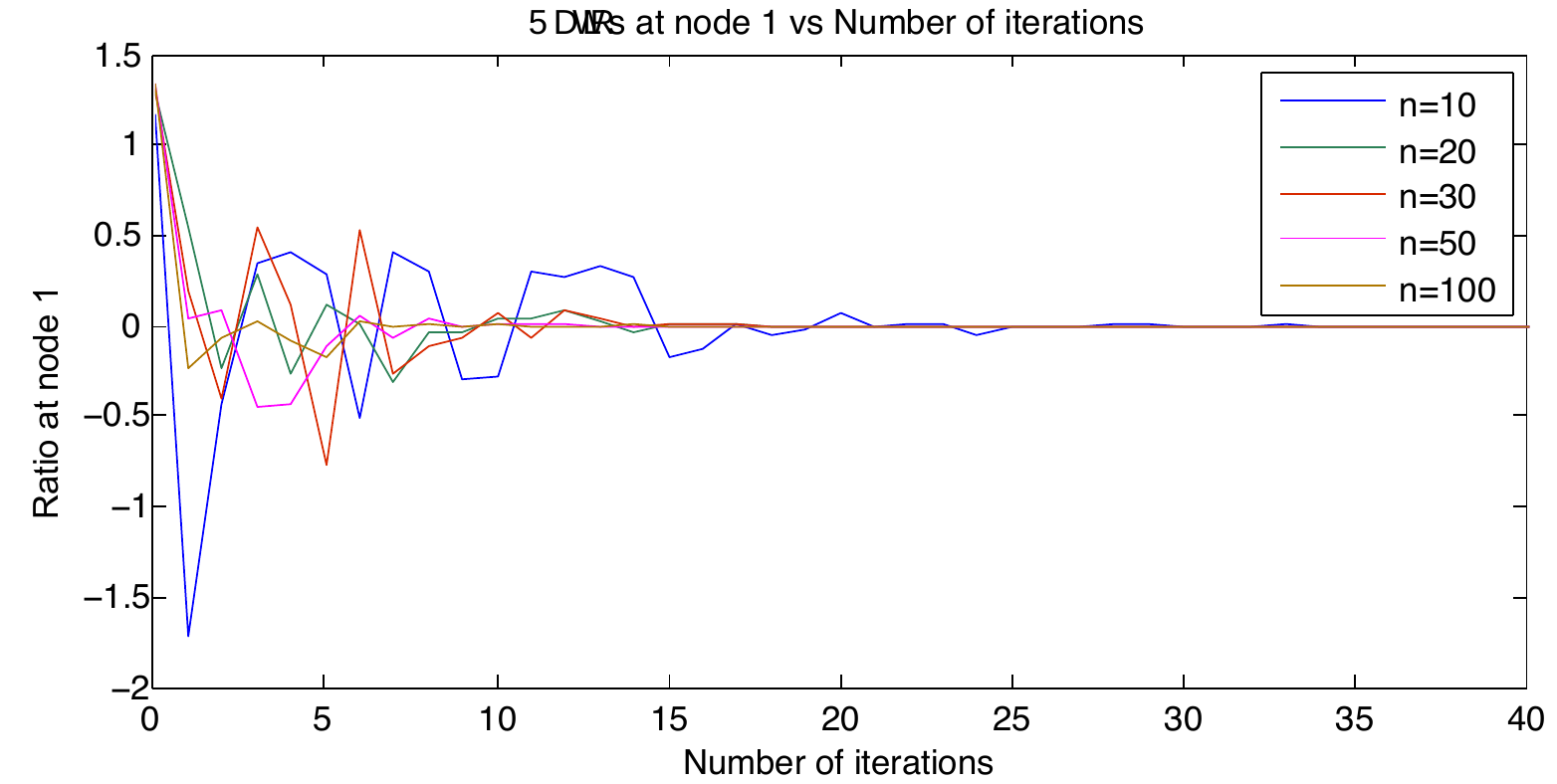}
\label{fig:comparison-size}
}
\caption{Convergence of the ratio at node $1$ for different upper bound $\bar{\tau}$ on delays (left) and different network sizes (right). 
}
\label{fig:delayed_comparenode1}
\end{figure}

\begin{remark}
In \cite{1986:Mitra, 2005:lei_fang}, the following update formula is suggested
\begin{align}
x_j[k+1]=p_{jj}' x_{j}[k] + \sum_{v_i \in \mathcal{N}^{-}_j}p_{ji}' x_{i}[k-d_{ji}[k]] \; , \quad k=0, 1, 2, \ldots
\label{eq:dsdelay}
\end{align}
where $x[0]=y_0$, the weights $p'_{ji}$ form a doubly stochastic weight matrix $P'=[p'_{ji}]$ and $d_{ji}[k]$ is chosen so that node $v_j$ uses in its update the most recently seen value from node $v_i$ (i.e., $d_{ji}[k] = \min_{\tau_{ji}[k-t] = t, 0 \leq t \leq \bar{\tau}} \{ t \}$). Since the weight matrix $P'=[p_{ji}']\in \mathbb{R}_+^{n\times n}$ is doubly stochastic, we know that {\em in the absence of delays} the iteration in \eqref{eq:dsdelay} would reach asymptotic average consensus. The iteration also reaches consensus in the presence of delays (regardless of the delays introduced, as long as they are bounded \cite{1986:Mitra}), but not necessarily to the exact average of the initial values. The value the nodes converge to depends on the specific delays that are introduced during the execution of the iteration.
\end{remark}


%
%
%
%
\section{Handling Changing Interconnections}\label{switching}

In this section, we extend the previous setting to include time-varying communication links (in addition to bounded delays on each link). We assume that we have a time-varying digraph, in which the set of nodes is fixed but the communication links can change, i.e., at time step $k$ the interconnections between components in the multi-component system are captured by a digraph $\mathcal{G}[k]=(\mathcal{V}, \mathcal{E}[k])$. For the analysis below, we let $\mathcal{\bar{G}}=\{\mathcal{G}_1, \mathcal{G}_2, \ldots, \mathcal{G}_\nu\}$, $\nu \leq 2^{n^2-n}$, be the set of all possible digraphs\footnote{Each of the $n$ nodes may be connected (out-going link) with up to $(n-1)$ other nodes. As a result, we have $n(n-1)$ possible links, each of which can be either present or not. Hence, we have $2^{n(n-1)}$ possible graph combinations. Of course, depending on the underlying application, some of these interconnection topologies may be unrealizable.} defined for a given set of nodes $\mathcal{V}$. We start our analysis by considering the simplest case where we have changing interconnection topology without delays and assuming that each node knows its (instantaneous) out-degree at that particular time instant. We start with the assumptions below (some of these assumptions are relaxed later on).
\begin{assum}
For the analysis below, the interconnection topology is described by a sequence of digraphs $\mathcal{G}[1]$, $\mathcal{G}[2]$, $\ldots$, $\mathcal{G}[k]$, $\ldots$, of the form $\mathcal{G}[k] = (\mathcal{V}, \mathcal{E}[k])$.
\begin{enumerate}
\item[(B1)] At each time instant $k$, each node $v_j$ knows the number of nodes receiving its message (i.e., the number of its out-neighbors $\mathcal{D}^{+}_{j}[k]$). \label{B11}
\item[(B2)] There exist no delays in the delivery of messages.
\item[(B3)] We can find an infinite sequence of times $t_0, t_1, \ldots, t_m, ...$, where $t_0=0, 0 <t_{m+1}-t_m \leq \ell < \infty$, with the property that for any $m \in \mathbb{Z}_{+}$ the union of graphs $\mathcal{G}[t_m], \mathcal{G}[t_m +1], \ldots, \mathcal{G}[t_{m+1}-1]$ is strongly connected.
\end{enumerate}
\end{assum}

\begin{remark}
Assumption (B1) requires that the transmitting node knows the number of nodes receiving its messages at each time instant. In an undirected graph setting, this is not too difficult; in a digraph setting, this is not as straightforward but there are many ways in which knowledge of the out-degree might be possible. For example, there can be an acknowledgement signal (ACK) via a \emph{distress signal} (special tone in a control slot or some separate control channel) sent at higher power than normal so that it is received by transmitters in its vicinity \cite{2000:bambos_channel}. Knowledge of the out-degree is also possible if the nodes periodically perform checks to determine the number of their out-neighbors (e.g., by periodically transmitting the distress signals mentioned above). As we discuss later, at the cost of little additional complexity, the nodes can also handle situations where they learn their out-degree with some delay. Assumption (B2) is made to keep things simple and is relaxed later. Assumption (B3) stems from the fact that we require that there exists paths between any pair of nodes infinitely often. 
\end{remark}

In its general form, each node updates its information state according to the following relation:
\begin{align}\label{eq:1_2}
x_{j}[k+1] =p_{jj}[k] x_{j}[k] + \sum_{v_i \in \mathcal{N}^{-}_j[k]} x_{j\leftarrow i}[k] \; , \quad k = 0, 1, \ldots
\end{align}
\noindent where $x_{j\leftarrow i}[k]\triangleq p_{ji}[k]  x_{i}[k]$ is the information sent from node $v_i$ to node $v_j$ at time step $k$, and $x_{j}[0] \in \mathbb{R}$ is the initial state of node $v_j$. Since the out-degree is known at the transmitting node, each node $v_j$ can easily set the (positive) weights $p_{lj}[k]=\frac{1}{1+\mathcal{D}^+_j[k]}$ for $v_l \in \mathcal{N}^+_j[k] \cup \{ v_j \}$ (this choice satisfies $\sum_{l=1}^n p_{lj}[k] = 1$ for all $v_j\in \mathcal{V}$ but more generally, as in the case of a fixed topology, each node only needs to ensure that the weights on its out-going links are positive and sum to unity).
Note that unspecified weights in $P[k]$ are set to zero and correspond to pairs of nodes $(v_l,v_j)$ that are not connected at time step $k$, i.e., $p_{lj}[k]=0$, for all $(v_l, v_j) \notin \mathcal{E}[k]$, $l \neq j$. If we let $x[k]=(x_1[k] \ \ x_2[k] \ \ \ldots  \ \  x_n[k] )^T$ and $P[k] = [p_{ji}[k] ] \in \mathbb{R}_{+}^{n\times n}$ then \eqref{eq:1_2} can be written in matrix form as
$
x[k+1] =P[k] x[k] ,
$
where $x[0] = (x_1[0] \ \ x_2[0] \ \ \ldots  \ \  x_n[0] )^T \equiv x_0^T$. Note that, with the specific choice of matrix $P[k]$ (based on the out-degree of each node as described above), the matrices $P[k]$ are column stochastic and have strictly positive elements on their diagonal. This fact will be important in our proof later on which utilizes Theorem~\ref{Wolfowitz} on a particular set of matrices.

\begin{remark} \label{re:distress}
\noindent
Throughout the operation of the algorithm, communication links can be initiated or terminated by either (a) the receiving node, or (b) the transmitting node. Possible communication protocols to perform these tasks are described briefly below. \\
(a) When node $v_l$ wants to receive messages from node $v_j$ (e.g., because it is in the neighborhood of $v_j$),  it can send a distress signal to pass this request to $v_j$ (alternatively, node $v_l$ can send the message to node $v_j$ using some path in the digraph or using some sort of flooding scheme). When node $v_j$ receives the request from $v_l$, it sends an acknowledgement packet (directly to node $v_l$) and the communication link is initiated. In practice, this might not necessarily require node $v_j$ to transmit a separate package to node $v_l$ (e.g., in a wireless broadcast setting) or to transmit at a higher power (e.g., if $v_l$ is already in its range); however, it does imply that node $v_j$ will adjust its self-weight and the weights $p_{lj}$, $v_l \in \mathcal{N}_j^+$, on the links to its out-neighbors in order to ensure that column stochasticity is preserved. If, on the other hand, node $v_l$ wants to terminate the communication link, it sends (or broadcasts if there exists a single communication channel and the message cannot be specifically directed to node $v_j$) a distress signal destined for node $v_j$ (alternatively, it can use a flooding-like strategy via the paths in the digraph); as soon as node $v_l$ receives an acknowledgement from node $v_j$ along with the latest message with values for the last update, then the link can be terminated. If node $v_l$ does not receive the acknowledgement message from node $v_j$ the link remains active. \\
(b) Note that if the transmitting node $v_j$ wants to terminate a communication link to node $v_l$, it is enough to simply initiate such a request to node $v_l$ (since a direct link is available).
%
\end{remark}

\begin{lem}\label{our_lemma2}
Consider a sequence of graphs of the form $\mathcal{G}[k]=(\mathcal{V}, \mathcal{E}[k])$), $k=0, 1, 2, ...$, such that there exists an infinite sequence of time instants $t_0, t_1, \ldots, t_m, \ldots$, where $t_0=0, 0 <t_{m+1}-t_m \leq \ell < \infty$, $m \in \mathbb{Z}_{+}$, with the property that for any $m \in \mathbb{Z}_{+}$ the union of graphs $\mathcal{G}[t_m], \mathcal{G}[t_m +1], \ldots, \mathcal{G}[t_{m+1}-1]$ is strongly connected. Let $y_j[k]$, $\forall v_j\in\mathcal{V}$, be the result of iteration~\eqref{eq:1_2}
with $p_{lj}[k] = \frac{1}{1 + \mathcal{D}_j^+[k]}$ for $v_l \in \mathcal{N}_j^+[k] \cup \{ v_j \}$ (zeros otherwise) and initial conditions $y[0]=y_0$, and let $z_j[k], \ \forall v_j\in\mathcal{V}$, be the result of iteration \eqref{eq:1_2} with $p_{lj}[k] = \frac{1}{1 + \mathcal{D}_j^+[k]}$ for $v_l \in \mathcal{N}_j^+[k] \cup \{ v_j \}$ (zeros otherwise) and with initial condition $z[0]=\mathbb{1}$. Then, the solution to the average consensus problem in the presence of dynamically changing topologies can be obtained as
$
\displaystyle \lim_{k\rightarrow \infty}\mu_j[k] = \frac{\sum_{v_\ell \in \mathcal{V}} y_0(\ell)}{|\mathcal{V}|} \; , \; \forall v_j \in \mathcal{V} \; ,
$
where $\displaystyle \mu_j[k]=\frac{y_j[k]}{z_j[k]}$.
\end{lem}

\begin{proof}
Let $\overline{P}_{t_{m+1}-t_m}\triangleq P[t_{m+1}-1]P[t_{m+1}-2]\ldots P[t_{m}]$. Since the union of graphs from time instant $t_m$ until $t_{m+1}-1$, i.e., the set of graphs $\mathcal{G}[t_m], \mathcal{G}[t_m +1], \ldots, \mathcal{G}[t_{m+1}-1]$, is strongly connected and each matrix involved in the product has strictly positive elements on the diagonals, matrix $\overline{P}_{t_{m+1}-t_m}$ is SIA for $m \in \mathbb{Z}_+$. Furthermore, products of matrices of the form $\overline{P}_{t_{m+1}-t_m}$ are SIA (since each such matrix has positive elements on its diagonal, the product of such matrices will have a positive entry at its $(j,i)$ position if at least one of the matrices has a positive element at its $(j,i)$ position; hence, the zero/nonzero structure of the product will correspond to strongly connected graph). Hence, according to Theorem~\ref{Wolfowitz}, for any $\epsilon > 0$, there exist a finite integer $\nu(\epsilon) \in \mathbb{N}$, such that a finite word $W$ given by the product of a collection of $\nu$ stochastic matrices of the form $\overline{P}_{t_{m+1}-t_m}$ has all of its columns approximately the same, i.e.,
$\overline{P}_{t_{k+\nu}-t_{k+\nu-1}}\ldots \overline{P}_{t_{k+2}-t_{k+1}}\overline{P}_{t_{k+1}-t_k} = c_{W_\nu} \mathbb{1}^{T}  + E$,
where $c_{W_\nu}$ is a nonnegative column vector and matrix $E$ has entries that are bounded in absolute value by $\epsilon/2$. From this point onwards, the proof continues as in the proof of Theorem~\ref{the:delay}.
\end{proof}

\subsection{Changing interconnection topology with communication delays}

\begin{assum}
In the presence of delays, we make the following extra assumption:
\begin{enumerate}
\item[(C1)] There exists a finite $\bar{\tau}$ that uniformly bounds the delay terms, i.e. $\tau_{ji}[k] \leq \overline{\tau}_{ji} \leq \overline{\tau}$; this is the same as in assumption ({A2}). 
\end{enumerate}
\end{assum}

\noindent In this case, each node updates its information state according to the following iteration:
\begin{align}\label{eq:S1}
x_{j}[k+1] &=p_{jj}[k]x_{j}[k] +  \sum_{r=0}^{\bar{\tau} }\sum_{v_i \in \mathcal{N}^{-}_j[k-r]} x_{j\leftarrow i}[k-r]I_{k-r,ji}[r] \; , \quad k = 0, 1, 2, \ldots
\end{align}
\noindent where $x_{j\leftarrow i}[k-r] \triangleq p_{ji}[k-r] x_{i}[k-r]$ is the value sent from node $v_i$ to node $v_j$ at time step $k-r$ that occurs delay $r$, $x_{j}[0] \in \mathbb{R}$ is the initial value of node $v_j$, and the values $p_{ji}[k]\geq 0$ depend on the topology of the graph at time $k$.

To handle delays in a network of $n=|\mathcal{V}|$ nodes, we introduce $\bar{\tau} n$ nodes (for a total of $(\bar{\tau}+1)n$ nodes) so that we can write
\begin{align}\label{eq:S3}
\overline{x}[k+1]=\overline{P}[k]\overline{x}[k] \; ,
\end{align}
where (as before)
\begin{align}\label{AkS}{\small
\overline{P}[k] \triangleq
\left( \begin{array}{ccccc}
    P_0[k] & I_{n\times n} & 0 & \cdots & 0 \\
    P_1[k] & 0 & I _{n\times n}&  \cdots &  0 \\
    \vdots & \vdots & \vdots & \ddots & \vdots \\
    P_{\bar{\tau}-1}[k] & 0 & 0 & \cdots & I_{n\times n} \\
    P_{\bar{\tau}}[k] & 0 & 0 & \cdots & 0 \\
  \end{array} \right),}
\end{align}
with
\begin{align*}
\overline{x}[k]&=\left( x^T[k] \ \ x^{(1)}[k] \ \ldots \ x^{(\bar{\tau})}[k]  \right)^T , \\
x^{(r)}[k]&=\left(x_{1}^{(r)}[k]  \ \ldots \  x_{n}^{(r)}[k]   \right) , \quad r=1, 2, \ldots \bar{\tau}.
\end{align*}
As before, $P_0[k], P_1[k], \ldots, P_{\bar{\tau}}[k]$ are appropriately defined nonnegative matrices, such that
\begin{align*}
P[k]=\sum_{r=0}^{\bar{\tau}}P_r[k],
\end{align*}
i.e., the sum $P[k]$ of all the nonnegative matrices $P_r[k]$, $r \in \{0,1,2, \ldots, \bar{\tau}\}$, gives the weights of the zero-delay interconnection topology at time instant $k$. The difference from the case when only delays are present in the network is that the interconnection topology is dynamically changing and the weights at each time instant might differ (mathematically, this means that the left matrix in the above equation changes with $k$). The proposed protocol is able to asymptotically reach average consensus, as stated in Lemma \ref{our_lemma3} below. The proof is similar to the proof for delays with no changes in the interconnection topology and is omitted.

\begin{lem}\label{our_lemma3}
Consider a sequence of graphs of the form $\mathcal{G}[k]=(\mathcal{V}, \mathcal{E}[k])$), $k=0, 1, 2, ...$, such that there exists an infinite sequence of time instants $t_0, t_1, \ldots, t_m, \ldots$, where $t_0=0, 0 <t_{m+1}-t_m \leq \ell < \infty$, $m \in \mathbb{Z}_{+}$, with the property that for any $m \in \mathbb{Z}_{+}$ the union of graphs $\mathcal{G}[t_m], \mathcal{G}[t_m +1], \ldots, \mathcal{G}[t_{m+1}-1]$ is strongly connected. Let $y_j[k]$ for all $v_j \in\mathcal{V}$ be the result of iteration \eqref{eq:S1}
with $p_{lj}[k] = \frac{1}{1 + \mathcal{D}_j^+[k]}$ for $v_l \in \mathcal{N}_j^+[k] \cup \{ v_j \}$ (zeros otherwise) and initial conditions $y[0]=y_0$, and let $z_j[k], \ \forall v_j\in\mathcal{V}$, be the result of iteration \eqref{eq:1_2} with $p_{lj}[k] = \frac{1}{1 + \mathcal{D}_j^+[k]}$ for $v_l \in \mathcal{N}_j^+[k] \cup \{ v_j \}$ (zeros otherwise) and with initial condition $z[0]=\mathbb{1}$. The indicator function $I_{k,ji}$ captures the bounded delay $\tau_{ji}[k]$ on link $(v_j, v_i)$ at iteration $k$ (as defined in \eqref{eq:indicatorfunction}, $\tau_{ji}[k] \leq \bar{\tau}$). Then, the solution to the average consensus problem can be asymptotically obtained as
$
\displaystyle \lim_{k\rightarrow \infty} \mu_j[k]=\frac{\sum_{v_l\in \mathcal{V}} y_0(l)}{|\mathcal{V}|} \; , \; \forall v_j \in \mathcal{V} \; 
$,
where $\mu_j[k]=\frac{y_j[k]}{z_j[k]}$.
\end{lem}

We now discuss the case in which a node, say node $v_j$, receives an indication (e.g., an acknowledgement message) that one of its out-neighbors, say $v_l \in \mathcal{N}^+_j[k-1]$, no longer receives its transmissions. In other words, node $v_l \notin \mathcal{N}^+_j[k]$ but node $v_j$ finds out about it with some bounded delay that we denote by $T_{lj}[k]$. Such bounded delays could arise from communication protocols in a variety of ways, e.g., when using periodic acknowledgement signals like the distress signals discussed in Remark~\ref{re:distress}. Another way for node $v_j$ to discover that its out-degree has changed is by using acknowledgement signals that arrive at node $v_j$ through paths in the network that connect the out-neighbors of node $v_j$ to node $v_j$ (note that a direct link between an out-neighbor of node $v_j$ and node $v_j$ may not necessarily exist in a digraph).

We use $T_{lj}[k] \leq \overline{T} < \infty$ to denote the time it takes for node $v_j$ to realize that node $v_l$ is no longer in the set $\mathcal{N}_j^+[k]$. The problem that node $v_j$ has to address at time $k+T_{lj}[k]$ when it realizes that, at time steps $k$, $k+1$, $\dots$, $k+T_{lj}[k-1]$, it erroneously assumed an out-degree $\mathcal{D}_{je}^+$ that (supposing no other changes) was greater than the true out-degree $\mathcal{D}_{j}^+$, is that the weighted message from node $v_j$ that was not eventually conveyed to the out-neighbor $v_l$ (because the link was terminated) needs to be accounted for. The most straightforward way to reconcile this is to add these values back to node $v_j$. This can be done easily {\em as long as node $v_j$ keeps track of the messages it has recently transmitted ---within the last $\overline{T}$ steps--- to its out-neighbors}. Note, however, that node $v_j$ has to track these messages for each of its perceived out-neighbors, because if more than one out-neighbors terminate their links with $v_j$ (possibly at different time steps), then node $v_j$ must know what needs to be added back to its own value for each such former out-neighbor. One way to think about this in terms of the augmented digraph is that node $v_j$ adds, for each of its perceived out-neighbors,  $\overline{T}$ ``virtual" nodes that loop back to itself. These virtual nodes essentially keep track of the values that have been sent to each out-neighbor in the last $\overline{T}$ steps. The following example discusses this issue in more detail.

Consider, for example, node $v_j$ with two out-neighbors ($\mathcal{D}_{j}^+=2$) shown in Figure~\ref{fig:exampledelayedACK_0}. Suppose that the maximum delay required for an acknowledgement signal from any out-neighbor of node $v_j$ is $2$ (i.e., $\overline{T}=2$). Then, the model for the part of the network consisting of node $v_j$ and its out-neighbors $v_{l_1}, v_{l_2}$ is as shown in Figure~\ref{fig:exampledelayedACK_0}. Suppose now that  out-neighbor $v_{l_2}$ terminates the link $v_{l_2}\leftarrow v_j$ and node $v_j$ receives an ACK with delay $2$. This means that node $v_j$ erroneously considered an out-degree of $2$ (instead of $1$) for the last $2$ updates. In this model, the message is passed through the two extra ``virtual" nodes (added in a self-loop), allowing us to loop the weighted message back to node $v_j$ (see Figure~\ref{fig:exampledelayedACK_2}). Therefore, node $v_j$ is able to recover the lost values (sent to node $v_{l_2}$ that is no longer an out-neighbor). It is worth pointing out at this point that, unlike the previous case of fixed interconnections with bounded delays, the virtual nodes are no longer simply a question of modeling; in fact, in this example, we have to ensure that node $v_j$ essentially implements the functionality of the \emph{virtual buffers} (by remembering the messages that it has sent to its out-neighbors.

\begin{figure}[h]
\subfigure[A node $v_j$ and its two out-neighbors. The maximum delay required for an acknowledgement signal from the out-neighbors of node $v_j$ is $2$.]
{
\includegraphics[width=0.45\columnwidth]{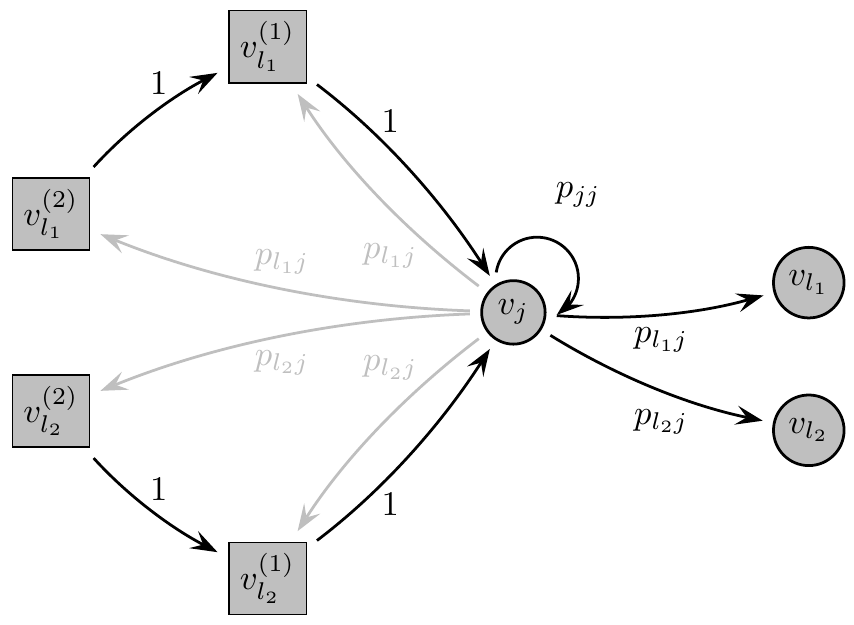}
\label{fig:exampledelayedACK_0}
}
\hfill
\subfigure[A model in which node $v_j$ directs the weighted messages of the links that no longer exist as delayed information to itself. In this example, node $v_{l_2}$ terminated the link and the ACK arrived at node $v_j$ with maximum delay $2$.]
{
\includegraphics[width=0.45\columnwidth]{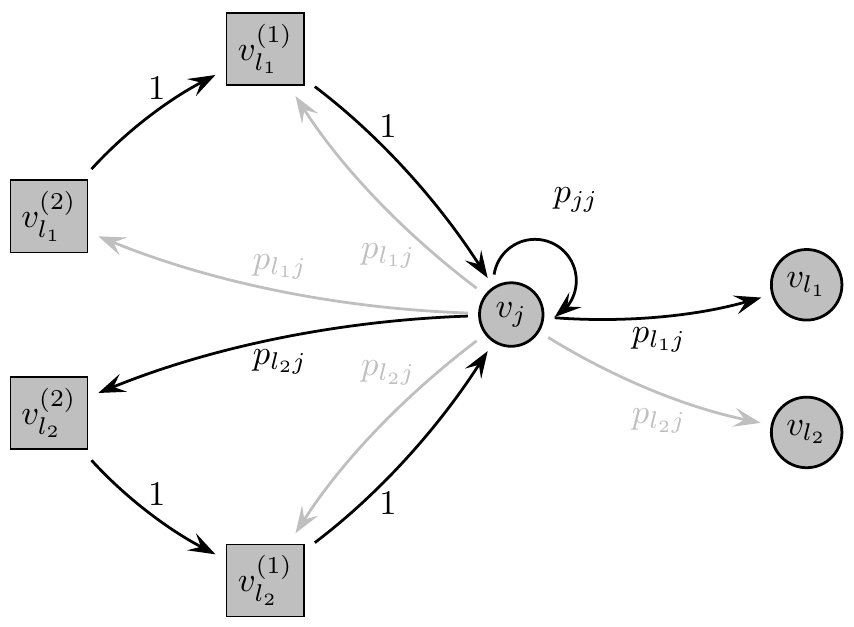}
\label{fig:exampledelayedACK_2}
}
\caption{The weighted messages from node $v_j$ that were not conveyed to the out-neighbor $v_{l_2}$ (because the link was terminated) are added back to the value of node $v_j$.
}
\label{fig:delayedACK}
\end{figure}

\begin{exam}
Consider the simple network with $3$ nodes in Figure~\ref{fig:exampledelayedACK_3} and the following scenario: at instant $k_1$ the network of the three nodes has no delays or interconnection topology changes; at time instant $k_2$, link $v_3\leftarrow v_1$ terminates, and an acknowledgement is sent to node $v_1$ via node $v_2$, from which there exists a link to node $v_1$. In addition, any message from $v_2$ to $v_1$ can be delayed by at most $1$ iteration.
\begin{figure}[h]
\begin{center}
\includegraphics[width=0.45\columnwidth]{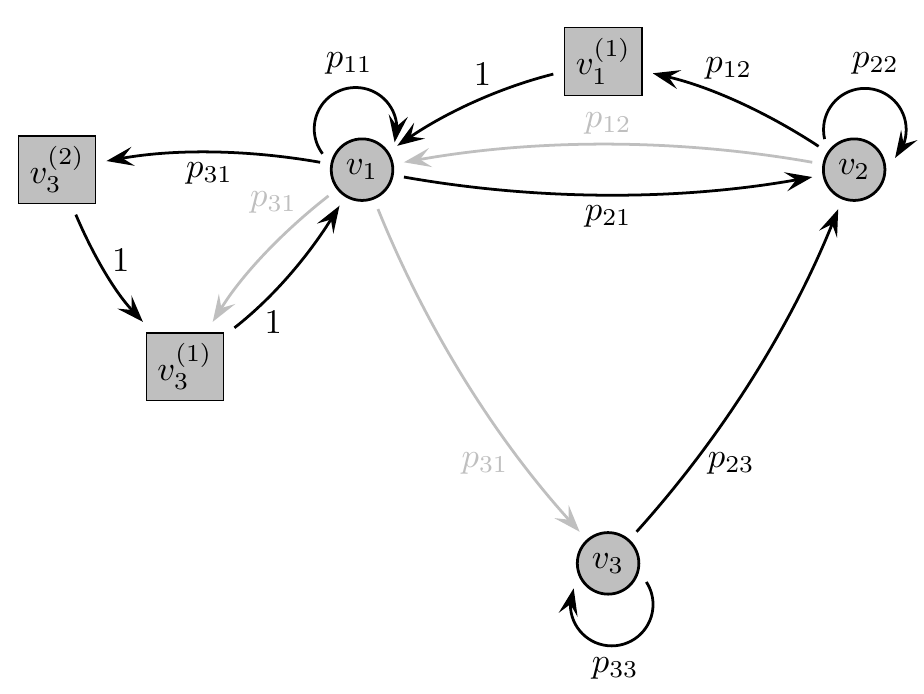}
\caption{A model in which link $v_3\leftarrow v_1$ terminates, and an acknowledgement is send to node $v_1$ via node $v_2$, from which there exists a direct link to node $v_1$. In addition, the message from $v_2$ to $v_1$ is delayed by $1$. }
\label{fig:exampledelayedACK_3}
\end{center}
\end{figure}
For simplicity, we assume that the delays in all other links are zero and the connection with link $v_2$ cannot be lost; hence, no additional loops need to be inserted. Taking $\overline{x}[k] =(x_1[k]  \ \  x_2[k]  \ \  x_3[k] \ \ x_{1}^{(1)}[k] \ \ x_{3}^{(1)}[k]  \ \  x_{3}^{(2)}[k] )^T$, the matrix representations at time instances $k_1$ and $k_2$ are captured by 
\begin{align*}{\small
\overline{P}[k_1]=
\left( \begin{array}{cccccc}
    p_{11} & p_{12} & 0           & 1 & 1 & 0 \\
    p_{21} & p_{22} & p_{23} & 0 & 0 & 0 \\
    p_{31} & 0           & p_{33} & 0 & 0 & 0 \\
    0           & 0           & 0          & 0 & 0 & 0 \\
    0           & 0           & 0          & 0 & 0 & 1 \\
    0           & 0           & 0          & 0 & 0 & 0 
  \end{array} \right),
\ 
\overline{P}[k_2]=
\left( \begin{array}{cccccc}
    p_{11} & 0           & 0           & 1 & 1 & 0 \\
    p_{21} & p_{22} & p_{23} & 0 & 0 & 0 \\
    0           & 0           & p_{33} & 0 & 0 & 0 \\ 
    0           & p_{12} & 0          & 0 & 0 & 0 \\
    0           & 0           & 0          & 0 & 0 & 1 \\
    p_{31} & 0           & 0          & 0 & 0 & 0 
  \end{array} \right).}
\end{align*}
\end{exam}
In the general case, in a network of $n=\mathcal{V}$ nodes, we introduce $\max(\overline{\tau}n,\overline{T}n)$ nodes (for a total of $(\max(\overline{\tau}n,\overline{T}n) +n)$ nodes) and we proceed as in \eqref{eq:S3}. 

\begin{remark}
There are also cases in which the transmitting node $v_j$ may not have knowledge of its out-degree at time instant $k$. Such situations can also be handled if, at each time instant $k$, node $v_j$ (i) knows the number of nodes with which it has established a communication link in the past and have not been officially terminated yet, and (ii) is able to multicast a table of values to each of these out-neighbors. One way to do this is to employ the communication protocol proposed in \cite{2012:Alejandro-Christoforos} where, at each time instant, each node $v_j$ broadcasts its own state (as updated via the iterations in equation \eqref{eq:1_2}), as well as the sum of all the values, called the \emph{total mass} in \cite{2012:Alejandro-Christoforos}, that have been broadcasted to each neighboring node $v_l \in \mathcal{N}_j^{+}$ so far. If, for any reason, some messages are lost (dropped) or the communication link disappears for some time-period, the total mass will enable the receiving out-neighbor to retrieve the information of the lost messages, with some time-delay. Thus, even though the communication links may not be reliable and can even change, the problem boils down to dealing with delayed information (as in Section~\ref{delays}). Note, however, that each node $v_j$ needs to keep track of its own current state, the total mass transmitted to each neighboring node $v_l \in \mathcal{N}_j^{+}[k]$ (the total mass can be different for each node $v_l$ due to, for example, newly established communication links), and the total mass received from each neighboring node $v_i \in \mathcal{N}_j^{-}[k]$ that transmits information to node $v_j$.  Since different information might need to be transmitted to each node $v_l$ at each time instant $k$, node $v_j$ is required to broadcast a table of values with entries for each receiving node.
\end{remark}

\begin{exam}
We illustrate how the algorithm operates via a small network of six nodes. Each node $v_j$ chooses its self-weight and the weight of its outgoing links at each time instant $k$ to be $(1+\mathcal{D}_j^{+}[k])^{-1}$ (such that the sum of all weights $p_{lj}[k]$, $v_l \in \mathcal{N}_j^+[k] \cup \{ v_j \}$, assigned by each node $v_j$ to links to its out-neighbors at time step $k$ is equal to $1$). First, suppose the nodes experience only changes in interconnection topology but no delays. When each node updates its information state $x_j[k]$ using equation \eqref{eq:1_2}, the information state for the whole network is given by
$
x[k+1]=P[k] x[k]
$,
where $P[k]$ depends on the links present at time instant $k$. 
For example, at time instants $k=k_1$ and $k=k_2$, whose interconnection topologies are captured by the graphs in Figures~\ref{directedgraph1} and~\ref{directedgraph2}, matrices $P[k_1]$ and $P[k_2]$, respectively, are given by
\begin{align*}{\small
P[k_1]=
\left( \begin{array}{cccccc}
    1/4  & 1/2 & 0    & 0    & 1/2  & 0 \\
    1/4  & 1/2 & 1/4 & 0    & 0     & 0 \\
    1/4  &  0   & 1/4 & 0    & 0     & 0 \\
    1/4  &  0   & 1/4 & 1/2 & 0    & 1/3 \\
    0     &  0   & 0     & 0    & 1/2 & 1/3 \\
    0     &  0   & 1/4 & 1/2 & 0     & 1/3 \\
  \end{array} \right), \
  P[k_2]=
\left( \begin{array}{cccccc}
    1/3  & 1/3 & 0    & 0    & 1/2  & 0 \\
    1/3  & 1/3 & 1/5 & 0    & 0     & 0 \\
    1/3  &  0   & 1/5 & 0    & 0     & 0 \\
    0     & 1/3 & 1/5 & 1/2 & 0    & 0 \\
    0     &  0   & 1/5 & 0    & 1/2 & 1/2 \\
    0     &  0   & 1/5 & 1/2 & 0    & 1/2 \\
  \end{array} \right).}
\end{align*}

\begin{figure}[h]
\subfigure[Connections and weights at instant $k_1$.]
{
\includegraphics[width=0.45\columnwidth]{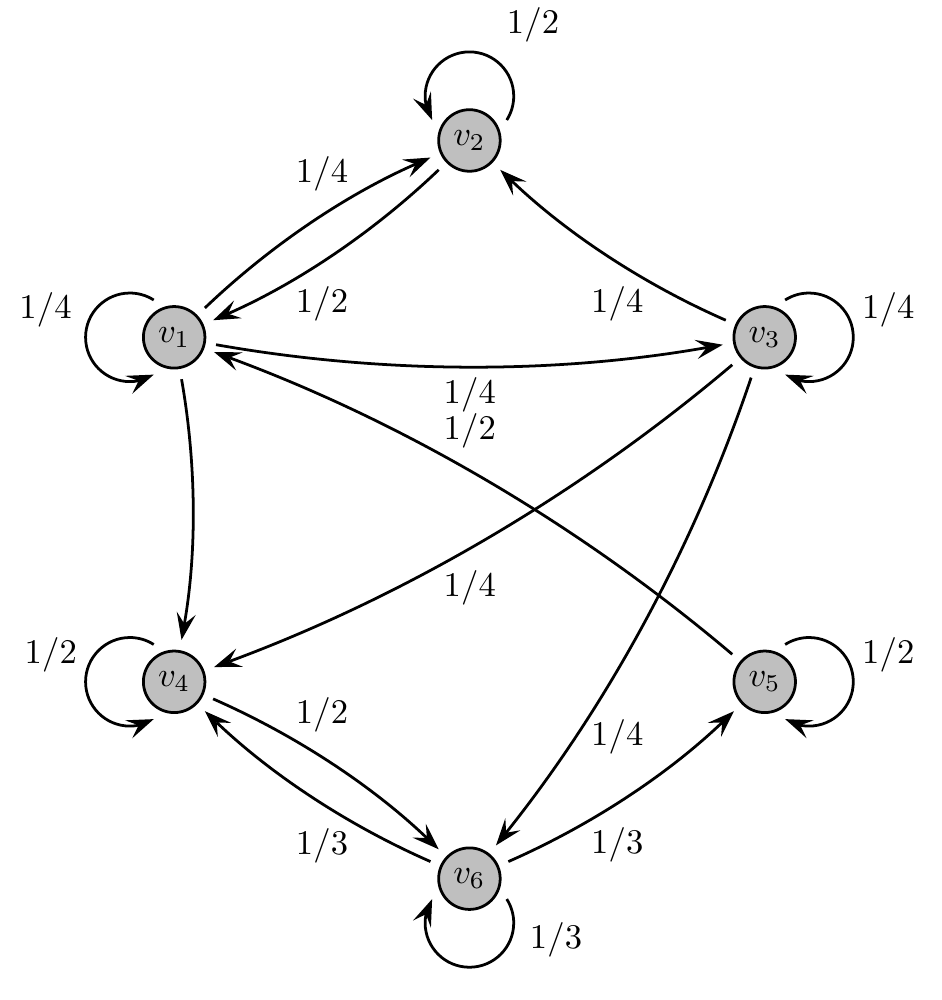}
\label{directedgraph1}
}
\hfill
\subfigure[Connections and weights at instant $k_2$.]
{
\includegraphics[width=0.45\columnwidth]{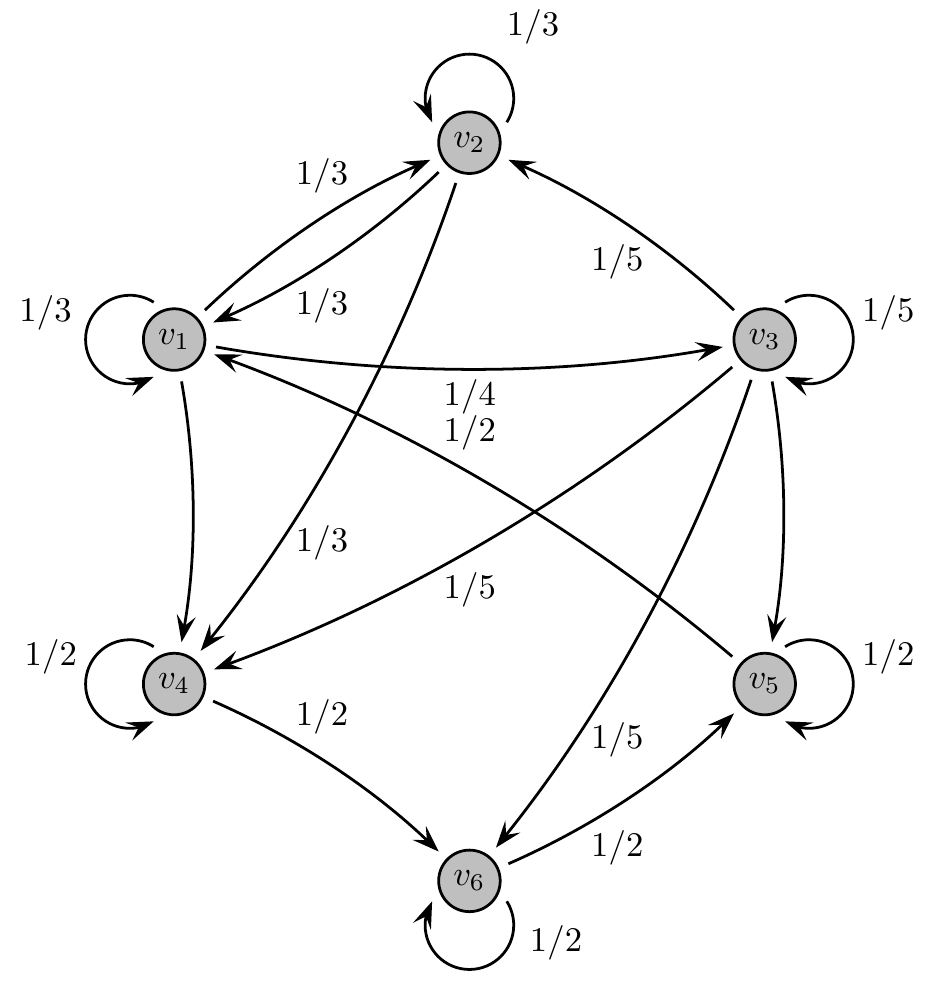}
\label{directedgraph2}
}
\caption{A network of six nodes, where each node $v_j$ chooses its self-weight and the weight of the links to its out-neighbors to be $(1+\mathcal{D}_j^{+}[k])^{-1}$.
}
\label{fig:directednet}
\end{figure}

We use twice the update formula \eqref{eq:1_2} with initial conditions $y[0]=(-1 \ \  1 \ \  2 \ \ 3 \ \ 4 \ \ 3)^T$ and $z[0]=\mathbb{1}^T$ respectively. In our simulations, we generate at each iteration a new random graph with six nodes that includes a directed link $(v_j, v_i)$ from node $v_i$ to node $v_j$ ($v_i, v_j \in V$, $v_i \neq v_j$) with some probability $p$ independently between different links. Note that once the graph is chosen at iteration $k$, the update matrix $P[k]$ will be column stochastic. A realization of the ratios at each node is shown in~Figure \ref{directed_s1}; in this case, the average is~$2$. When delays are present with maximum delay $\bar{\tau}=5$ we use the update formula \eqref{eq:S1} with the same initial conditions 
and we observe that the system again converges to the exact average (as shown in Figure~\ref{directed_s2}), but with a slower convergence.
\begin{figure}[h]
\subfigure[Ratios at each node for changing interconnection topology and no delays.]
{
\includegraphics[width=0.45\columnwidth]{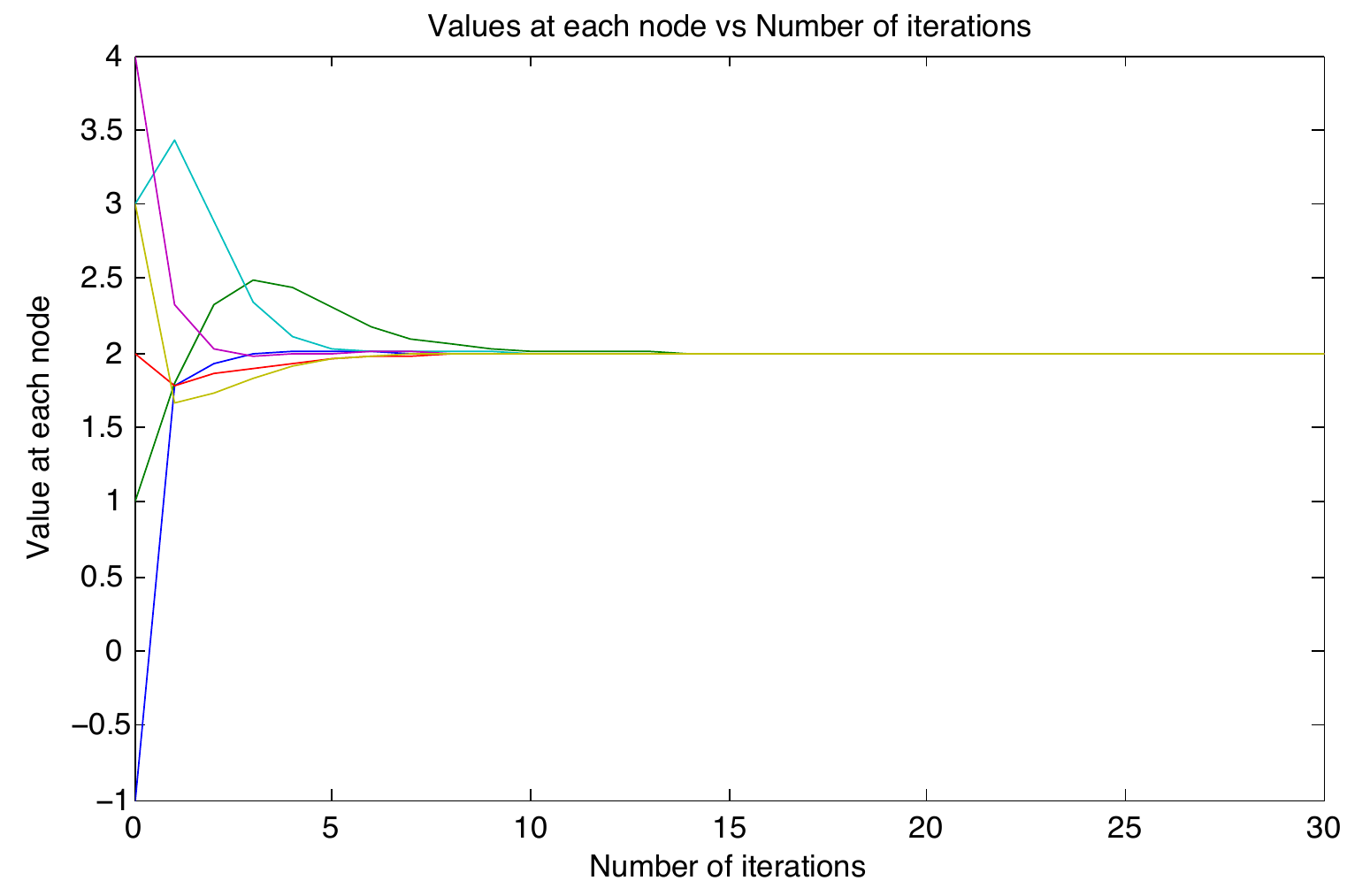}
\label{directed_s1}
}
\hfill
\subfigure[Ratios at each node for changing interconnection topology with delays.]
{
\includegraphics[width=0.45\columnwidth]{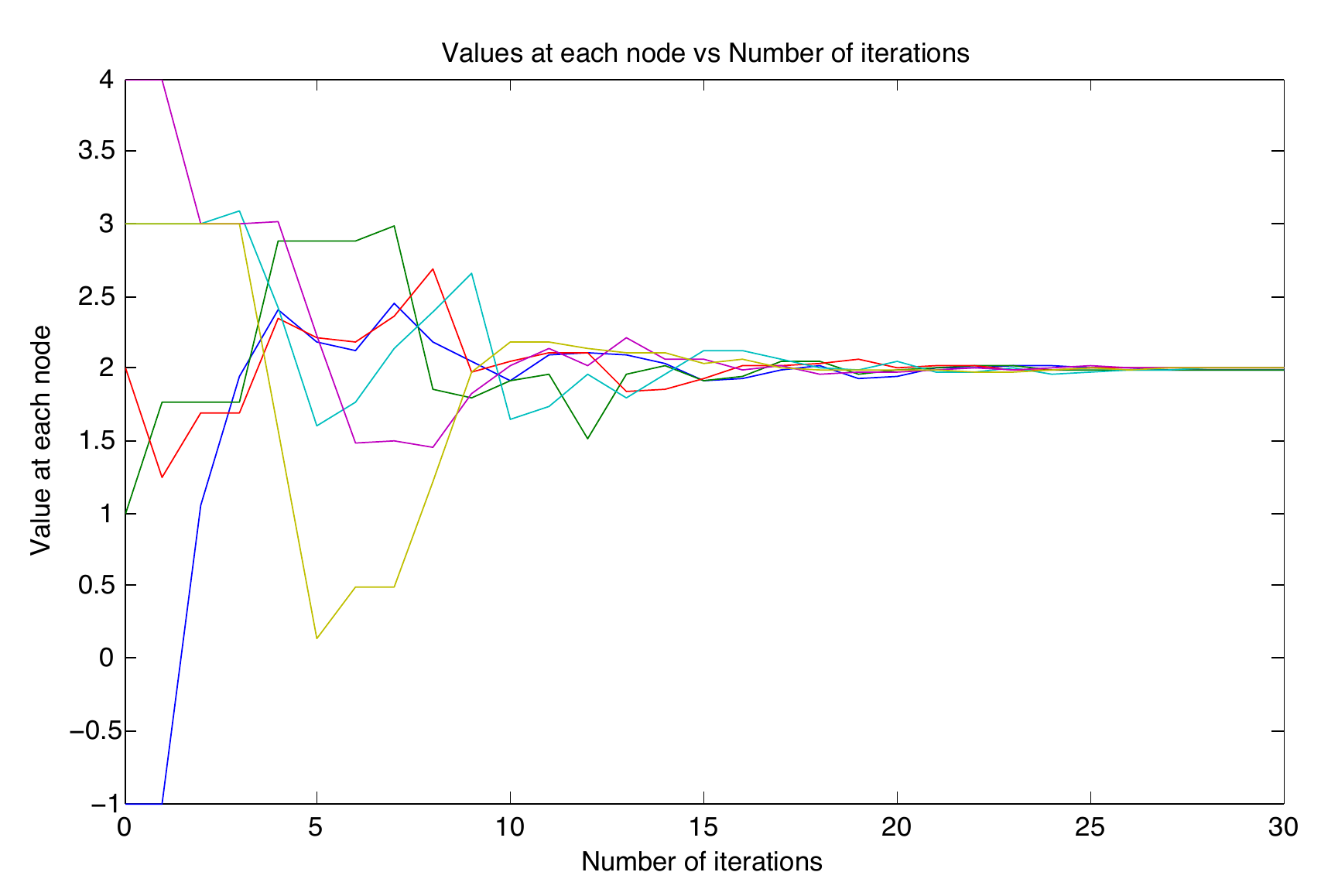}
\label{directed_s2}
}
\caption{We use twice the update formula \eqref{eq:1_2} with initial conditions $y[0]=(-1 \ \  1 \ \  2 \ \ 3 \ \ 4 \ \ 3)^T$ and $z[0]=\mathbb{1}^T$ respectively, and plot the ratio $y_j[k]/z_j[k]$ for each node $v_j$ under changing interconnection topology but no delays (left).  When delays are present with maximum delay $\bar{\tau}=5$ we use the update formula \eqref{eq:S1} with the same initial conditions and observe that the ratios again converge to the average, but with a slower convergence (right).
}
\label{directed_ss}
\end{figure}

\end{exam}

\begin{remark}
In many settings, it might be more natural for the communication protocol to allow the receiver to set the weights of the incoming values, since it is easier for each receiving node to know from which (and how many) nodes it has received a message. Indeed, consensus in multi-agent systems in digraphs in the presence of changing interconnection topology and time-varying delays have been studied in\cite{2006:XiaoWang}, which provided sufficient conditions for the multi-agent system to reach consensus when using a protocol that relies on a single iteration and uses weights that form a row stochastic matrix (compared to the two iterations and weights that form column stochastic matrices proposed in this paper). Event though \cite{2006:XiaoWang} reaches consensus, it does not necessarily reach consensus to the exact average of the initial values of the nodes; in fact, the value to which this approach (and others that rely on a single iteration and weights that form row stochastic matrices) converge depends on the delay magnitude and profile.

Technically, our approach relies on {\em weak convergence} of the backward product of column stochastic matrices whereas \cite{2006:XiaoWang} and others rely on {\em strong convergence} of the backward product of row stochastic matrices (which is equivalent to a forward product of column stochastic matrices\footnote{Note that weak and strong convergence are equivalent for forward products of column stochastic matrices (see Theorem 4.17 in the book of Seneta ``Non-negative Matrices and Markov Chains'').}). This means that in approaches that depend on row stochastic matrices the product of matrices converges to a rank one matrix (which would have to be equal to $\frac{1}{n}\mathbb{1}\mathbb{1}^T$ if convergence to the average of the initial values is desirable); on the contrary, in the proposed approach we do not have convergence of the product but we have instead weak convergence to rank one matrices. This means that for a large number of iterations, each matrix product gets closer to a rank one matrix but this rank one matrix is not necessarily the same for each iteration step. However, by running two iterations and focusing on the ratio of the two iteration values, we exploit weak convergence and are able to obtain the exact average.
\end{remark}

\section{CONCLUSIONS}\label{conclusions}

In this paper, we studied distributed strategies for a discrete-time networked system to reach asymptotic average consensus in the presence of time-delays and dynamically changing topologies. By assuming that nodes in the multi-agent system have knowledge of their out-degree (i.e., the number of nodes to which they send information to) and by modeling the time-delays using an augmented graphical model, we have shown that our proposed discrete-time strategy reaches asymptotic average consensus in a distributed fashion, in the presence of dynamically changing interconnection topology for whatever the realization of delays, as long as they are bounded and the union graph of the graph topologies over consecutive time intervals forms a strongly connected graph infinitely often. 


\bibliographystyle{IEEEtran}
\bibliography{bib_allerton}

\end{document}